\documentclass[11pt]{amsart}
\usepackage[utf8]{inputenc}
\usepackage{microtype}
\usepackage{amssymb,amsmath}
\usepackage[braket, qm]{qcircuit}
\usepackage{graphicx}
\usepackage{color}
\usepackage{tikz}
\usepackage{pgfplots}
\pgfplotsset{compat=1.18}
\usepgfplotslibrary{fillbetween}
\usepackage{enumitem}
\usepackage{algorithm}
\usepackage{setspace}
\usepackage{algpseudocode}
\usepackage{mathtools}
\usepackage{xfrac}

\usepackage{accents}
\usepackage{multicol} 
\usepackage{multirow}
\usepackage{soul}
\usepackage{subcaption}
\usepackage{booktabs}
\usepackage{array}
\usepackage{comment}
\newcolumntype{L}[1]{>{\raggedright\let\newline\\\arraybackslash\hspace{0pt}}m{#1}}
\newcolumntype{C}[1]{>{\centering\let\newline\\\arraybackslash\hspace{0pt}}m{#1}}
\newcolumntype{R}[1]{>{\raggedleft\let\newline\\\arraybackslash\hspace{0pt}}m{#1}}
\usepackage{bbold}
\usepackage{standalone}

\usepackage[hyphens]{url}
\usepackage[hidelinks]{hyperref}
\usepackage[nameinlink,noabbrev]{cleveref}

\usetikzlibrary{cd,pgfplots.statistics}
\newtheorem{theorem}{Theorem}
\newtheorem{proposition}[theorem]{Proposition}

\newtheorem{lemma}[theorem]{Lemma}
\theoremstyle{definition}
\newtheorem{definition}[theorem]{Definition}

\theoremstyle{remark}
\newtheorem{remark}[theorem]{Remark}

\Crefname{assumption}{Assumption}{Assumptions}
\numberwithin{theorem}{section}
\numberwithin{equation}{section}
\numberwithin{table}{section}
\numberwithin{figure}{section}
\definecolor{myBlue}{RGB}{30,144,255} 
\definecolor{myGreen}{RGB}{69,169,0} 
\definecolor{myRed}{RGB}{165,12,42} 
\definecolor{myOrange}{RGB}{225,92,22} 
\definecolor{mycolor0}{rgb}{0.12156862745098,0.466666666666667,0.705882352941177} 
\definecolor{mycolor1}{rgb}{0.00000,0.44700,0.74100}
\definecolor{mycolor2}{rgb}{0.85000,0.32500,0.09800}
\definecolor{mycolor3}{rgb}{0.49400,0.18400,0.55600}
\definecolor{mycolor4}{rgb}{0.92900,0.69400,0.12500}
\definecolor{mycolor5}{rgb}{0.46600,0.67400,0.18800}
\definecolor{mycolor6}{rgb}{0.30100,0.74500,0.93300}
\definecolor{mycolor7}{rgb}{0.63500,0.07800,0.18400}

\newcommand{\delete}[1]{ }

\DeclareFontFamily{U}{matha}{\hyphenchar\font45}
\DeclareFontShape{U}{matha}{m}{n}{
	<-6> matha5 <6-7> matha6 <7-8> matha7
	<8-9> matha8 <9-10> matha9
	<10-12> matha10 <12-> matha12
}{}
\DeclareSymbolFont{matha}{U}{matha}{m}{n}

\DeclareFontFamily{U}{mathx}{\hyphenchar\font45}
\DeclareFontShape{U}{mathx}{m}{n}{
	<-6> mathx5 <6-7> mathx6 <7-8> mathx7
	<8-9> mathx8 <9-10> mathx9
	<10-12> mathx10 <12-> mathx12
}{}
\DeclareSymbolFont{mathx}{U}{mathx}{m}{n}

\DeclareMathDelimiter{\vvvert} {0}{matha}{"7E}{mathx}{"17}%

\newcommand{\N}{\mathbb{N}}

\newcommand{\R}{\mathbb{R}}
\newcommand{\C}{\mathbb{C}}

\DeclareMathOperator{\poly}{poly}
\DeclareMathOperator{\SWAP}{SWAP}
\DeclareMathOperator{\Id}{Id}

\newcommand{\V}{\mathcal{V}}

\newcommand{\Q}{\mathcal{Q}}
\newcommand{\dx}{\, \mathrm{d}x}
\newcommand{\bigO}{\mathcal{O}}
\newcommand{\tol}{\mathtt{tol}}
\newcommand{\CNOT}[1]{\operatorname{C}_{#1}\!\operatorname{NOT}}
\newcommand{\keff}{\kappa_\text{eff}}

\allowdisplaybreaks
\begin{document}
\title[Quantum Realization of the Finite Element Method]{Quantum Realization of the Finite Element Method}
\author[]{M.~Deiml$^\dagger$, D.~Peterseim$^{\dagger*}$}
\address{${}^{\dagger}$ Institute of Mathematics, University of Augsburg, Universit\"atsstr.~12a, 86159 Augsburg, Germany}
\address{${}^{*}$ Centre for Advanced Analytics and Predictive Sciences (CAAPS), University of Augsburg, Universit\"atsstr. 12a, 86159 Augsburg, Germany}
\thanks{The work of D. Peterseim is part of a project that has received funding from the European Research Council (ERC) under the European Union’s Horizon 2020 research and innovation programme (Grant agreement No. 865751 – RandomMultiScales).}
\email{\{matthias.deiml,daniel.peterseim\}@uni-a.de}
%
\date{\today}

%
%
\begin{abstract}
This paper presents a quantum algorithm for the solution of prototypical second-order linear elliptic partial differential equations discretized by $d$-linear finite elements on Cartesian grids of a bounded $d$-dimensional domain. An essential step in the construction is a BPX preconditioner, which transforms the linear system into a sufficiently well-conditioned one, making it amenable to quantum computation. We provide a constructive proof demonstrating that, for any fixed dimension, our quantum algorithm can compute suitable functionals of the solution to a given tolerance $\tol$ with an optimal complexity of order $\tol^{-1}$ up to logarithmic terms, significantly improving over existing approaches. Notably, this approach does not rely on the regularity of the solution and achieves a quantum advantage over classical solvers already in two dimensions, whereas prior quantum methods required at least four dimensions for asymptotic benefits. We further detail the design and implementation of a quantum circuit capable of executing our algorithm, present simulator results, and report numerical experiments on current quantum hardware, confirming the feasibility of preconditioned finite element methods for near-term quantum computing.
\end{abstract}
%
\maketitle

\section{Introduction}
Quantum computers have the potential to provide an exponential speed-up over classical computing paradigms, a prospect that holds particular promise for the field of computational mathematics. Among the most compelling applications of this quantum advantage is the solution of partial differential equations (PDEs) \cite{MP16}, especially in high dimensions. The ability to solve PDEs more efficiently could improve numerical simulation in numerous fields, including engineering and physics.

A prime example within the vast landscape of PDEs is the steady-state diffusion equation
\begin{equation}\label{eq:modelproblem}
-\operatorname{div}(A\nabla u) = f,
\end{equation}
posed in a bounded domain $D \subset \R^d$, $d \in \N$, accompanied by homogeneous Dirichlet conditions $u = 0$ at the boundary of the domain $\partial D$. This equation serves as a building block for understanding various processes, from heat conduction and material diffusion for $d \le 3$ to high-dimensional problems in quantum chemistry where $d \gg 3$, underscoring the broad implications of advancing its solution methods.

The conventional approach to discretizing the steady-state diffusion equation \eqref{eq:modelproblem} utilizes the Galerkin finite element method. In cases of axis-aligned domains, such as the unit cube~$D = [0,1]^d$, the method of choice often involves continuous piecewise $d$-linear~(Q1) finite elements on a Cartesian grid composed of $n\times\ldots\times n=n^d$ cells. Even under regularity assumptions on the symmetric and positive definite diffusion coefficient~$A$ and the right-hand side function~$f$ that ensure the $H^2(D)$-regularity of the solution~$u=u_{A,f}$, the grid resolution $n$ must scale inversely with the tolerance $\tol > 0$ to reliably push the approximation error below this threshold. Consequently, achieving the desired tolerance~$\tol$ requires computational complexity that increases as $\tol^{-d}$. This exponential growth of complexity with respect to the physical dimension $d$ -- a phenomenon often referred to as the curse of dimensionality -- epitomizes the limitations of traditional PDE solution algorithms in classical computing systems, especially under conditions of moderate regularity~\cite{Bungartz_Griebel_2004}.

Quantum computing offers promising ways to overcome the computational hurdles posed by high-dimensional PDEs. Among the notable strategies is \emph{Schrö\-di\-nger\-iza\-tion}, which converts the original problem into a higher-dimensional Schrödinger equation solvable in theory with quantum efficiency \cite{PhysRevA.108.032603,HJZ23,HJLZ24}. Another approach taken in \cite{CPP+13, Ber14, BCOW17, CL20} is to treat the linear systems derived from PDE discretization with quantum linear system algorithms \cite{HHL09, Amb10, BCK15, CKS17, GSLW19, SSO19, LT20, AL22}. Similar methods are being explored in the context of variational quantum computing for finite element discretizations \cite{TLDE23, SWRK+23}. 
Despite their potential for logarithmic complexity in the problem size, all methods face significant challenges, most importantly the poor conditioning associated with conventional finite difference or finite element discretizations, which grows quadratically as $\tol^{-2}$ with decreasing tolerance under $H^2(D)$-regularity assumptions. It directly affects the efficiency of quantum linear system solvers, whose complexity grows linearly with the condition number and, thus, again quadratically with $\tol^{-1}$, giving a total runtime that is proportional to $\tol^{-3}$ in general. This implies that a quantum advantage with these approaches is only possible where $d > 3$, which represents a major obstacle to a direct quantum adaptation of standard computational PDE techniques in the near future. Furthermore, it does not seem to offer an advantage for highly relevant classes of challenging PDEs. More importantly, the practical use of quantum hardware is currently limited by the depth of the resulting quantum circuits, which must scale at least linearly with the condition number, adding another layer of complexity to their realization. Quantum methods with an improved runtime rely either on a periodic structure \cite{TAWL21, LOC24} or higher regularity \cite{CLO21} of the problem.

The challenges highlighted above underscore the urgent need to overcome the limitations of current quantum computing approaches to solving PDEs. This leads us directly to the critical role of preconditioning in the quantum computing context. In classical computing, preconditioning is a well-established practice aimed at accelerating iterative solvers by improving the conditioning of the problem. In quantum computing, however, the importance of preconditioning extends beyond mere acceleration. It becomes essential for maintaining the scale of the amplitude of the solution state and preventing the amplification of input errors, a phenomenon similar to error amplification in classical computations with finite-precision arithmetic. 

The straightforward application of conventional preconditioners, designed for classical systems to approximate the inverse system matrix efficiently, is not feasible. These operations could introduce ill-conditioned behavior, exacerbating the challenges in quantum computing.  Previous attempts to repurpose preconditioning as a means to improve quantum computational runtimes \cite{CJS13, BNWA23a, SX18} have generally had limited success. A notable exception is the methodology introduced in \cite{TAWL21}, demonstrating an efficient resolution for the Laplace equation under very specific conditions: constant coefficients and periodic boundary conditions. However, this approach's reliance on near-diagonalization of the system of equations by a Fourier transform prevents its application to a wider range of PDEs beyond periodicity structures.

Given the challenges that ill-conditioning poses to quantum solvers, this work aims to demonstrate the feasibility of the finite element method on quantum computers and its advantageous performance. Central to our approach is a re-evaluation of preconditioning strategies suitable for quantum environments, in which the classical BPX multilevel preconditioner \cite{BPX89} is adapted to the quantum context by integrating it with an appropriate stiffness matrix factorization to avoid ill-conditioned operations. Up to logarithmic terms, this adaptation results in a practical algorithm with complexity $\bigO(\tol^{-1})$, which is asymptotically optimal as $\tol$ tends to zero under reasonable assumptions. For fixed dimension $d$, this complexity is proportional to that of its one-dimensional counterpart as $\tol\rightarrow 0$. However, the hidden constant in the $\bigO$ notation may be affected by the dimension $d$. In particular, our algorithm not only achieves the efficiency of \cite{TAWL21} for the periodic Laplacian, but also significantly improves the performance of the aforementioned algorithms by a factor of $\tol^{-2}$ in the regime of $H^2(D)$-regularity and already enables a quantum advantage for $d > 1$. The advantage increases for less regular problems.

In classical computational frameworks, achieving this optimized complexity often depends on strict regularity conditions, such as those applied in the use of high-order finite elements~\cite{MR1695813}. Conversely, the use of low-order finite elements, as discussed here, lends itself to efficient compression by sparse grids or low-rank tensor methods, which exploit the tensor product structure of the problem and its diffusion coefficient, as well as the discretization. We refer to \cite{Bungartz_Griebel_2004} and \cite{Bac23} for the corresponding comprehensive reviews. In fact, tensor methods are conceptually very closely related to quantum computing \cite{Vid03}. Our algorithmic solution utilizes factorizations of the discretized problem that have previously found application in this area \cite{BK20}, thus bridging efficient strategies on classical computers and the emerging quantum computing paradigm.

\subsection*{Notation}
We use the standard bra-ket notation to denote quantum states, which are normalized vectors in $\C^{2^n}$ for some $n \in \N$. For this, the one-qubit states $\ket{0}$ and $\ket{1}$ denote the vectors $(1, 0)^T$ and $(0, 1)^T$, respectively. Similarly, if $0 \le x \le 2^n-1$ is an unsigned $n$-bit integer, then $\ket{x}$ denotes the state of an $n$-qubit register storing said integer. More precisely, if $x$ has base $2$ representation $x_n \dots x_1$ with $x_j \in \{0, 1\}$ for $1 \le j \le n$, then
$$\ket{x} = \ket{x_1} \otimes \dots \otimes \ket{x_n} = (\delta_{jx})_{j=0}^{2^n-1}.$$
Chained kets or bras like $\ket{a}\ket{b}$ indicate the semantic partitioning of a state into registers and refer to the tensor product $\ket{a}\ket{b} = \ket{a} \otimes \ket{b}$.
Finally, for normalized vectors $v \in \C^{2^n}$ we sometimes write $\ket{v}$ instead of $v$. For any vector in ``ket'' notation, the corresponding ``bra'' means the conjugate transposed vector, i.e.\ $\bra{\bullet} = \ket{\bullet}^\dagger$. An optional subscript such as in $\ket{0}_1$ refers to the number of qubits $n$ in the qubit register.

Additionally we use the ``big O'' notation $\bigO(\bullet)$ for asymptotic comparison of algorithms. The hidden constants are independent of the discretization constants but may depend on the dimension $d$ and the diffusion coefficient $A$.
%
%
\section{Model problem}
\label{sec:preliminaries}
The model problem considered in this paper is a prototypical second-order linear partial differential equation \eqref{eq:modelproblem} with a homogeneous Dirichlet boundary condition on a bounded Lipschitz domain $D$ in $d$ spatial dimensions. For the sake of simplicity, we restrict ourselves to the particular choice of the $d$-dimensional unit cube $D := [0, 1]^d$. For the treatment of more general domains, see Remark~\ref{rem:general-domain}.

The solution space of the problem is the Sobolev space $H_0^1(D)$ and the associated bilinear form reads
\begin{equation}\label{eq:bilinear}
    a(u, v) = \int_D A\nabla u\cdot \nabla v\dx
\end{equation}
for $u,v\in H^1_0(D)$. The matrix-valued diffusion coefficient $A$ is a function of the spatial variable $x$ with measurable entries and is assumed to be symmetric positive definite. For simplicity, we assume that there are positive constants $\alpha$ and $\beta$ such that 
\begin{equation}\label{eq:coer}
    \alpha|\eta|^2\leq \eta\cdot A\eta\leq \beta |\eta|^2
\end{equation}
holds for all $\eta\in \R^d$ and almost all $x\in D$, where $|\eta|$ is the Euclidean norm of $\eta$. The condition \eqref{eq:coer} guarantees that the bilinear form \eqref{eq:bilinear} is an inner product on $H_0^1(D)$. It induces the energy norm $\|\cdot\|$, which is
equivalent to the canonical norm on that space. Under the condition \eqref{eq:coer}, the variational boundary value problem
\begin{equation}\label{eq:modelweak}
a(u, v) = f^*(v),\quad v \in H_0^1(D),
\end{equation}
has, according to the Lax-Milgram theorem, for all bounded linear functionals $f^*$ on $H_0^1(D)$ a unique solution $u\in H_0^1(D)$.

The standard way of approximating the model problem \eqref{eq:modelweak} in a wide range of applications is the Galerkin method, that is, restricting the variational problem to some finite-dimensional subspace of $H^1_0(D)$. The most popular choice is the finite element space of continuous piecewise polynomial functions with respect to some suitable subdivision of the domain $D$. For our choice $D = [0, 1]^d$, let $\mathcal{G}_h$ be the Cartesian grid of $D$ of width~$h=2^{-L}$ for some $L\in\N$, i.e.\ the subdivision of $D$ into $2^{dL}$ mutually distinct cubes of width~$h$. We then consider the $Q1$ finite element space $\V_h\subset H_0^1(D)$ of continuous piecewise $d$-linear functions on that grid vanishing on the boundary $\partial D$. The unique solution~$u_h\in \V_h$ of the discrete variational problem
\begin{equation} \label{eq:weak-discrete}
a(u_h, v_h) = f^*(v_h)_{L^2},\quad v_h \in \V_h,
\end{equation}
is the best approximation of $u\in H_0^1(D)$ in the energy norm
\begin{equation}\label{eq:error}
    \|u-u_h\|=\min_{v_h\in\V_h}\|u-v_h\|.
\end{equation}
If the problem data are sufficiently smooth in the sense that the diffusion coefficient $A$ is weakly differentiable with $\nabla A\in L^\infty(D,\R^{d\times d})$ and $f\in L^2(D)$ then $u\in H^2(D)$ and standard interpolation error estimates yield
\begin{equation}\label{eq:errorh}
\|u-u_h\|\leq C_{A,f}h,
\end{equation}
where $C_{A,f}$ only depends on $\alpha$, $\beta$, $\|\nabla A\|_{L^\infty(D,\R^{d\times d})}$ and $\|f\|_{L^2(D))}$.
The characteristic linear rate of convergence with respect to $h$ is optimal under the given assumption of regularity. To ensure an error below the tolerance $\tol > 0$ thus requires the choice~$h \approx \tol$ in general, which implies that the number of degrees of freedom, and hence the computational cost on a classical computer, grows at least as $\tol^{-d}$. Notably, for general diffusion coefficients $A\in L^\infty(D)$, the convergence of the finite element method can become arbitrarily slow \cite{MR1648351},  further amplifying the complexity growth as the dimension~$d$ increases. Only under stronger regularity assumptions can this unfavorable growth with the dimension~$d$ be relaxed. Classical options include the use of higher order finite elements~\cite{MR1695813}. Its regularity requirements and convergence properties are analyzed even for the case of highly oscillatory multiscale coefficients in \cite{MR3022017}. This approach has also been explored for quantum computers~\cite{CLO21}. Alternatively, the above choice of low-order finite elements can be efficiently compressed by sparse grids \cite{Bungartz_Griebel_2004} or low-rank tensor methods \cite{Bac23} based on the tensor product structure of the discretization introduced above for the special case of tensor product diffusion coefficients.

A choice of basis functions $\Lambda_1, \dots, \Lambda_N$ of $\V_h$ transforms the discrete variational problem \eqref{eq:weak-discrete} into a system of $N\coloneqq\dim \V_h=(2^L-1)^d$ linear equations
\[Sc = r\]
for the coefficients $c=(c_1,\ldots,c_N)^T$ of $u_h=\sum_{j=1}^{N}c_j\Lambda_j$ in its basis representation. The system matrix $S \coloneqq (a(\Lambda_j,\Lambda_k))_{j,k=1}^N$ is given by the $a$-inner products of the basis functions, and the right-hand side vector $r \coloneqq ((f, \Lambda_j)_{L^2})_{j=1}^N$ contains the evaluation of the functional~$f^*$ in the basis functions. For the particular choice of hat (or Lagrange) basis functions $\Lambda_j$ associated with the inner vertices $x_j$ of $\mathcal{G}_h$ which are uniquely characterized by the Lagrange property $\Lambda_j(x_k)=\delta_{jk}$ with the Kronecker delta $\delta_{jk}$, the coefficient vector $c$ contains exactly the nodal values~$c_j=u_h(x_j)$. Furthermore, the matrix $S$ is sparse with at most $3^d$ nonzero entries per row and column. In practice, the exact evaluation of the integrals underlying the entries of $S$ and $r$ is not always possible and a suitable quadrature is required. This case of a so-called variational crime \cite{BS08} is well understood in the numerical analysis literature and will not be discussed in this manuscript. In particular, we implicitly assume that the diffusion coefficient $A$ is well approximated by a piecewise constant function on the grid~$\mathcal{G}_h$ without further mention.

\section{Encoding vectors and matrices on quantum computers}\label{sec:encoding}
As a first step towards realizing a quantum implementation of the finite element method, we will address the necessary processing of vectors and matrices on a quantum computer. This so-called encoding is critical to realizing the potential speed-up that quantum computing offers over traditional classical computing methods. For readers who are not familiar with the underlying basic principles of quantum computing, including qubits and quantum gates, we recommend the introductory article \cite{doi:10.1137/18M1170650}. For a more comprehensive mathematical overview, we refer to \cite{Lin22}.
Here, we give some background on quantum computing concepts in the context of linear algebra. Although these are mostly well-established results, we improve the precision of the notation of encodings where the dimension is not a power of $2$ and introduce the concept of \textit{subnormalization}, which is central to the findings of the following sections.

In quantum computing, a vector $c \in \C^{2^m}$ for $m\in\N$ is represented by a register of $m$ qubits. Using the aforementioned Bra-ket notation, the canonical basis vectors~$e_j \in \C^{2^m}$, where $(e_j)_k = \delta_{jk}$ for $j,k = 0, \ldots, 2^m - 1$, are denoted as $\ket{j}$. In practice, this representation corresponds to a quantum state in which the qubits in a register encode the integer~$j$ in binary, assigning one digit per qubit.
The linear combination of basis vectors, $c_0 e_0 + \dots + c_{2^m-1} e_{2^m-1}$, is realized as a superposition $\sum_j c_j \ket{j}$ of basis states. It is imperative that these states are normalized to ensure admissibility on a quantum computer, which is satisfied when the Euclidean norm of $c$ satisfies $|c| = 1$. Under this normalization condition, the resulting state is represented by $\ket{c}$.

A matrix $S$, on the other hand, is encoded not as a data structure stored in quantum memory, but as an operation that transforms a basis state $\ket{i}$ into a superposition corresponding to the matrix-vector product of $S$ and the basis vector represented by $\ket{i}$ (the $i$-th column of $S$ if indexing starts at zero).  This method relies on the inherently linear nature of quantum operations, ensuring that linear combinations of basis states are correctly mapped according to the action of the matrix. However, quantum transformations must also be unitary, i.e.\ norm-preserving. Encoding arbitrary linear transformations thus requires embedding them in unitary transformations acting on higher-dimensional spaces. 
Specifically, we use projection operators to precisely embed the matrix $S$ in a larger unitary operation $U_S$. These projections are encoded as generalized controlled NOT gates. We will use $U_S$ to refer to both the unitary matrix and its implementation interchangeably.
\begin{definition}[Projection as $\CNOT{\Pi}$ gate]
    Let $N, m \in \N$ with $N \le 2^m$ and $\Pi : \C^{2^m} \to \C^N$ be a projection, that is, a matrix with orthonormal rows. We call a gate of $m + 1$ qubits a $\CNOT{\Pi}$ gate if it flips the last bit if and only if the first $m$ bits represent a vector in the range of $\Pi$, i.e.~it performs the unitary operation 
\[ \Pi^\dagger\Pi \otimes (\op{0}{1} + \op{1}{0}) + (1 - \Pi^\dagger\Pi) \otimes \Id.\]
\end{definition}
The subsequent definition of block encodings of matrices uses such projection gates onto both the domain and the range of the matrices.

\begin{definition}[Block encoding of a matrix, normalization, and subnormalization]\label{d:block}
Let $N_1, N_2, m \in \N$ with $N_1,N_2 \le M \coloneqq 2^m$. Let $S \in \C^{N_2 \times N_1}$ and $\gamma \ge |S|$. Further, let $\Pi_1 \colon \C^M \to \C^{N_1}$ and $\Pi_2 \colon \C^M \to \C^{N_2}$ be projections. Consider a quantum algorithm implementing a unitary matrix $U_S \in \C^{M \times M}$ and implementations of a $\CNOT{\Pi_1}$ and a $\CNOT{\Pi_2}$ gate. We call the tuple $(\gamma, U_S, \CNOT{\Pi_1}, \CNOT{\Pi_2})$ a \textit{block-encoding} of $S$ if
\[ \gamma\Pi_2 U_S \Pi_1^\dagger = S. \]
For simplicity we will refer to the whole block encoding by $U_S$.
The scalar $\gamma(U_S) \coloneqq \gamma$ is called the \textit{normalization factor} and $\tilde \gamma(U_S) \coloneqq \gamma(U_S) / |S| \ge 1$ is called the \textit{subnormalization factor}, where $|\bullet|$ denotes the spectral norm.
\end{definition}

The above definition of block encoding differs slightly from the one in \cite{GSLW19,Lin22}, where the number of ancilla bits and an approximation error of the encoding are used as additional parameters, while the subnormalization factor is not introduced. In the present context, subnormalization is an important quantity that measures the extent to which the norm of a vector is preserved under the action of block encoding, with $\tilde \gamma(U_S) = 1$ being optimal. We will see later that the computational complexity of solving a linear system is proportional to $\tilde \gamma(U_S) \cdot \kappa(S)$ and not just to $\kappa(S)$. The normalization factor $\gamma(U_S)$, on the other hand, tracks the relationship between the encoding and the actual matrix $S$. Knowing this quantity is essential to correctly interpret the result of applying the encoding to a vector.

Unitary matrices are trivially encoded by a circuit that implements them. In this case, both $\Pi_1$ and $\Pi_2$ are identity matrices. However, such trivial encoding is generally efficient only for ``small'' unitary matrices, that is, those whose size does not depend on the size of the problem. To construct block encodings of larger and more complicated matrices, we will make use of the following calculus of block encodings. Therein, bounds for the normalization and subnormalization factors are quantified in terms of the condition number $$\kappa(A):=|A| \cdot |A^+| = \sigma_{\max}(A)/\sigma_{\min}(A)$$
of a matrix $A$ with maximal and minimal singular values $\sigma_{\max}(A)$ and $\sigma_{\min}(A)$. Here, $\bullet^+$ is the Moore-Penrose pseudoinverse, so $\kappa(A)$ is well defined even for rectangular matrices. Note that $\kappa(A)$ may differ from the ratio of the largest and smallest eigenvalue~$\lambda_{\max}(A)/\lambda_{\min}(A)$ in the case of square but non-symmetric matrices.
\begin{proposition}[Operations on block encoded matrices] \label{prop:block-encoding-ops}
Let $A \in \C^{m_A \times n_A}, B \in \C^{m_B \times n_B}$ be matrices with respective block encodings $U_A$ and $U_B$. Then block encodings with normalization $\gamma$ and subnormalization $\tilde \gamma$ can be constructed for the following operations:
\begin{subequations}
\begin{align}
A \otimes B&\quad\text{ with }\gamma = \gamma(U_A) \gamma(U_B)\text{ and }\tilde \gamma = \tilde \gamma(U_A)  \tilde \gamma(U_B), \label{eq:op-tensor} \\
A^\dagger &\quad\text{ with }\gamma = \gamma(U_A)\text{ and }\tilde \gamma = \tilde \gamma(U_A) \label{eq:op-transpose} \\
\begin{bmatrix}A & 0 \\ 0 & B \end{bmatrix}&\quad\text{ with }\gamma = \max\{\gamma(U_A), \gamma(U_B)\} \text{ and } \tilde \gamma \le \max\{\tilde \gamma(U_A), \tilde \gamma(U_B)\}, \label{eq:op-block-diagonal}\\[-2ex]
AB &\quad\text{ with }\gamma = \gamma(U_A)  \gamma(U_B) \text{ and }\tilde \gamma \le\begin{cases} \kappa(A)  \tilde \gamma(U_A)  \tilde \gamma(U_B) \text{ if } m_A \ge n_A,\\
\kappa(B)  \tilde \gamma(U_A)  \tilde \gamma(U_B) \text{ if } n_B \ge m_B,\end{cases}\label{eq:op-mutiply}\\
\begin{bmatrix}A & B \end{bmatrix}&\quad\text{ with } \gamma = \sqrt{\gamma(U_A)^2 + \gamma(U_B)^2} \text{ and }\tilde \gamma \le \sqrt{\tilde \gamma(U_A)^2 + \tilde \gamma(U_B)^2}, \label{eq:op-block-horizontal}\\
\mu_A A + \mu_B B &\quad \text{ for any }\mu_A, \mu_B \in \C \text{ with }\gamma = |\mu_A\gamma(U_A)| + |\mu_B\gamma(U_B)|,\label{eq:op-add}
\end{align}
\end{subequations}
provided that the dimensions are compatible in the sense of $\Pi_{B,2} = \Pi_{A,1}$ in \eqref{eq:op-mutiply}, $\Pi_{A,2} = \Pi_{B,2}$ in \eqref{eq:op-block-horizontal} and \eqref{eq:op-add}, and $\Pi_{A,1} = \Pi_{B,1}$  in \eqref{eq:op-add}. Without further assumptions, the subnormalization in \eqref{eq:op-add} can be unbounded.
\end{proposition}
The results for multiplication and linear combination are due to \cite{GSLW19}. A complete proof of the proposition along with concrete quantum circuits is given in the Appendix~\ref{app:proof}.
Note that in all cases, the subnormalization factor of the resulting block encoding can be calculated from Definition \ref{d:block}, as long as the spectral norm of the encoded matrix is known. This is especially useful for the addition of matrices, where an estimate of the spectral norm of the result requires further assumptions on the summands.

For efficient treatment of the diffusion coefficient, we will later take advantage of the fact that block diagonal matrices can be efficiently implemented for many cases.
\begin{proposition}[Block encoding of block diagonal matrices] \label{prop:block-encoding-diag}
Let $n, m, N, M \in \N$ with $N = 2^n$ and $M \le 2^m$. Let $\gamma > 0$ and $A \in NM\times NM$ be a $N \times N$ block diagonal matrix with blocks of size $M \times M$. Let $\Pi \colon \C^{2^m} \to \C^M$ be a projection with an implementation of $\CNOT{\Pi}$. Let $U_A$ be a quantum circuit operating on two qubit registers $\ket{j}_n, \ket{x}_m$ applying a block encoding of the $j$-th diagonal block $A_{jj}$ with projection $\Pi$ and normalization $\gamma$ on the register $\ket{x}$, that is, \[ \gamma(\bra{j} \otimes \Pi) U_A (\ket{j} \otimes \Pi) = A_{jj}.\]
Then we can construct a block encoding of $A$ with normalization $\gamma$. Its subnormalization is the minimum of the subnormalization over all blocks.
\end{proposition}
\begin{proof}
It is easy to see that $U_A$ is a block encoding of $A$ with projection $\Id_N \otimes \Pi$. To estimate the subnormalization simply observe that the spectral norm $|A|$ is the maximum of the spectral norms of all blocks.
\end{proof}
Block encodings can also be efficiently implemented directly for other classes of matrices, such as sparse matrices or matrices that can be written as a linear combination of unitaries (LCU) \cite{GSLW19}. However, it is clear that one cannot encode $k$ bits of information without using $\bigO(k)$ operations, i.e., we cannot expect to encode arbitrary matrices of size $N$ in fewer than $\bigO(N^2)$ operations. This applies analogously to vectors. Since we want to achieve runtimes faster than linear in $N$, this implies that the data must be available in some compressed way. In the context of finite elements, this means that we expect the diffusion coefficient entering the stiffness matrix $S$ and the right-hand side $\ket{r}$ to be given by subroutines that compute specific integrals of either function. For details see the Appendix \ref{app:implementation-vectors}. Note that the Cartesian meshes have a compressed representation through their mesh sizes, which determine the coordinates of all $\bigO(N)$ nodes and elements. If one wished to use more general meshes, it would still be necessary to describe them using information that scales at most polylogarithmically with $N$ to fully exploit the quantum potential.

\section{The Quantum Linear System Problem}
As outlined in Section \ref{sec:preliminaries} we want to use quantum computers to solve linear systems of the form
\[ Sc = r\]
arising from a particular finite element discretization of the partial differential equation~\eqref{eq:modelproblem}. Without loss of generality it is assumed that $|S| = 1$ and $|r| = 1$ in this section. In practice, this assumption is easily removed by an appropriate rescaling of the equation.

In the quantum context we assume that the right-hand side is given as an oracle. The oracle should prepare the corresponding superposition $\ket{r}$, that is, it maps the initial state~$\ket{0}$ to the state $\ket{r}$. This is explained in detail in Appendix~\ref{app:implementation-vectors}. In addition, we assume that the matrix $S$ is given as a block encoding, although there are other input models, such as oracles which compute the sparsity structure and the corresponding entries of the matrix. 
Since $|S^{-1}| = \kappa(S)$ might be larger than $1$ it is not possible in general to find a state $\ket{c}$ such that
\[ \ket{r} = (\bra{0} \otimes \Id)U_S \ket{0}\ket{c} = S\ket{c}.\]
Such a $\ket{c}$ may exceed $1$ in norm and thus may not be representable on a quantum computer. Instead, quantum solvers compute the normalized solution up to a tolerance as specified in the following definition.
\begin{definition}[Quantum linear system algorithm] \label{d:qlsa}
Consider an algorithm that takes as input a block encoding $U_S$ of a matrix $S$ and a unitary $U_r$ which prepares some state $\ket{r}$, meaning $U\ket{0} = \ket{r}$. We say that the algorithm solves the quantum linear system problem for \textit{effective condition} $\keff > 0$ and tolerance $\tol > 0$, if given $U_S$ with $\keff \ge \tilde \gamma(U_S)\kappa(S)$ it produces a state $\ket{c}$ approximating the normalized solution to $Sc = r$ such that 
\[ \bigg|\ket{c} - \frac{S^{-1}\ket{r}}{|S^{-1}\ket{r}}\bigg| < \tol. \]
The number of calls to each of $U_S$ and $U_r$ in the algorithm is denoted by the algorithm's \textit{ query complexity}. We also allow the algorithm to fail with a probability that is independent of $\kappa$ and $\tol$, in which case it should communicate the failure via an additional flag bit.
\end{definition}
There are several algorithms for solving linear algebraic systems on quantum computers. The first of these algorithms \cite{HHL09} constructs a state that encodes $\keff^{-1} S^{-1} \ket{c}$ and normalizes it by amplitude amplification, achieving a complexity of $\bigO(\keff^2\tol^{-1})$. Recently, it has been shown that it is possible to improve the dependence of the complexity bound on both $\keff$ and $\tol$, either by more sophisticated techniques such as variable-time amplitude amplification \cite{Amb10, BCOW17, CGJ19}, or by constructing the normalized state $\ket{c}$ directly, thus avoiding the amplification process \cite{AL22, LT20}. For a detailed comparison of these approaches, see \cite[Appendix A]{TAWL21}. Many of these algorithms achieve the runtime quantified in the following theorem and thus serve as a possible proof.
\begin{theorem}[Existence of fast quantum linear system algorithms] \label{thm:qlsa}
A solution of the quantum linear system problem in the sense of Definition~\ref{d:qlsa} can be computed with a query complexity  of $\bigO(\keff \poly \log(\keff\, \tol^{-1}))$.
\end{theorem}
Similarly, one can obtain an estimate of the norm of the solution $|S^{-1}\ket{r}|$ with complexity $\bigO(\tol^{-1}\keff\poly\log(\keff\,\tol^{-1}))$ \cite{CGJ19}. There is an extra factor of $\tol^{-1}$ in this runtime compared to preparing $\ket{c}$, but the norm estimate has to be computed only once, while one typically has to prepare the state $\ket{c}$ multiple times.

Specifically, once the solution $\ket{c}$ is computed, the reliable measurement of all its components requires at least $\bigO(N/\tol)$ samples of the state $\ket{c}$ and thus $\bigO(N/\tol)$ runs of the solver algorithm, eliminating any advantage over classical computers. A quantum advantage can be expected when computing partial information from $\ket{c}$ only. Typically, the goal is to measure some \textit{quantity of interest} of $\ket{c}$, which is computed either as a product $\melem{c}{M}{c}$ with a matrix $M$, for example using the \textit{Hadamard test} \cite{AJL06} in the case where $M$ is unitary, or as a linear functional 
$$c \mapsto \ip{m}{c}$$
represented by a vector $m \in \C^N$ using the \textit{Swap Test} \cite{BCWD01}. This requires an efficient construction of the state $\ket{m}$ or the unitary transformation $M$. Note that even to measure a quantity of interest up to the tolerance $\tol$, any solver algorithm must be run at least $\bigO(\tol^{-1})$ times, resulting in a total runtime of $\bigO(\keff\tol^{-1}\poly\log(\keff\tol^{-1}))$. This is not a specific phenomenon for linear systems. Estimating any quantity encoded in the amplitudes of a quantum state, a process called \textit{amplitude estimation}, requires at least $\bigO(\tol^{-1})$ measurements \cite{Md23}, and thus it is reasonable to assume that computing any information about the solution to \eqref{eq:modelproblem} using a quantum computer must be of at least linear complexity in $\tol^{-1}$. In the context of the finite element method, this is consistent with the results of \cite{MP16}, where it is shown that the runtime of the quantum algorithms to solve \eqref{eq:modelproblem} must be at least polynomial in $\tol^{-1}$. However, the optimal dependence on $\tol^{-1}$ is not achieved by trivially applying a quantum linear system algorithm to the linear system $Sc = r$. The condition number $\kappa(S) \in \bigO(h^{-2})$ introduces a multiplicative factor, which in the case of $H^2(D)$-regularity translates into $\tol^{-2}$, because $h\lesssim \tol$ is required to ensure a sufficiently small discretization error according to \eqref{eq:errorh}. We instead use \textit{preconditioning} to avoid suboptimal scaling of the algorithms caused by the increasing condition number of the stiffness matrix for smaller tolerances.
%
%
\section{Quantum Preconditioning}
\label{sec:quantum-preconditioning}
Starting with a discussion of its challenges and limitations, this section focuses on preconditioning as an essential strategy to improve the performance of quantum computing methods. The concept of preconditioning, well known in classical computing for its effectiveness in improving the efficiency of iterative solvers, faces unique challenges and opportunities when applied to quantum computing. Given the fundamental differences of the quantum paradigm, in particular the requirement for unitary operations, traditional preconditioning techniques must be reevaluated and adapted. This adaptation is critical to address the conditioning issues that significantly impact the feasibility and efficiency of quantum algorithms to solve PDEs.

\subsection{Challenges introduced by unitary transformation}
Preconditioning is a technique that transforms a problem into a more favorable form for numerical solution, typically by left- and right- preconditioning of the form $P_1SP_2y = P_1r$. We will focus on the case of \textit{operator preconditioning} \cite{Hiptmair}, where $P_1$ and $P_2$ represent discretizations of linear operators, denoted as $\mathcal{R}_1$ and $\mathcal{R}_2$. The classical application involves assembling the matrices $P_1$, $S$, and $P_2$, along with the vector $r$, to compute the right-hand side $P_1r$. Iterative solvers then utilize this vector and the matrices for solving the preconditioned linear equation, ultimately leading to the solution of the original problem via $c = P_2y$.

Initially, one might consider the direct implementation of classical preconditioning techniques in a quantum context. This would involve constructing a quantum state corresponding to the right-hand vector $\ket{r}$, finding block encodings for the matrices $P_1$, $S$, and $P_2$, and then computing the quantum equivalent of the matrix-vector product $\ket{P_2r}$. With the matrix product $P_1SP_2$ at hand, obtained through quantum-specific matrix-matrix multiplication  \eqref{eq:op-mutiply} from Proposition \ref{prop:block-encoding-ops}, one could then use a quantum linear system solver to generate a state $\ket{y}$. Ideally, this state $\ket{y}$, after further manipulation using the block encoding of $P_2$, would yield the solution to the original problem. However, this approach has a fundamental limitation inherent in quantum computing: the spectral norm of any matrix implemented on a quantum computer must be bounded from above by $1$. This limitation implies that multiplication by a matrix with a condition number $\kappa$ could, in the worst case, decrease the norm of the resulting vector by a factor of $\kappa$. Although the loss of scale can be somewhat reversed by a process called \textit{amplitude amplification} \cite{BH97, Gro98, Md23}, this is not possible without repeating all the steps leading to the obtained state at least $\bigO(\kappa)$ times. This requirement negates any potential efficiency gains, as the total runtime remains linearly dependent on $\kappa$.

Moreover, effective preconditioning, by definition, seeks to adjust the condition number of a problem. For preconditioners $P_1$ and $P_2$ to be effective, their total condition~$\kappa(P_1)\kappa(P_2)$ must be close to $\kappa(S)$. Thus, attempting to apply the typical matrix multiplication strategy~\eqref{eq:op-mutiply} to perform the sequential multiplication of the preconditioners and the matrix on a quantum state results in significant subnormalization of the quantum state, proportional to~$\kappa(S)$. 
This directly undermines the feasibility of straightforwardly applying classical preconditioning techniques in a quantum context, marking a key problem where quantum computing diverges from classical computing paradigms.

A similar problem exists when considering the influence of machine precision on classical computers on the accuracy of the result. Assembling all matrices and then using classical matrix preconditioning leads to a numerical error of the order of $\bigO(\kappa \epsilon)$, where $\epsilon$ is the precision of the machine. In contrast, directly assembling a discretization of $a(\mathcal{R}_1 \bullet, \mathcal{R}_2 \bullet)$ and $f^*(\mathcal{R}_1^* \bullet)$ -- an approach called \textit{full operator preconditioning}  \cite{MNU24} -- can avoid this dependence on $\kappa$. The same is possible on a quantum computer: by constructing a block encoding of $P_1SP_2$ and the state $\ket{P_1 r}$ directly, rather than by multiplication with the preconditioner, the problems described above can be avoided. 

For right preconditioning, the computation of the solution as $c = P_2y$ should also be avoided on quantum computers, while the measurement of $m^Tc$ for a vector $m$ may be feasible using the equivalence $m^Tc = m^TP_2y = (P_2^Tm)^Ty$. Consequently, we can obtain the result of a linear functional as long as we can efficiently construct $P_2^Tm$. However, it is unclear how a quadratic quantity of interest of the form $\melem{c}{M}{c}$, otherwise popular in quantum computing, can be computed efficiently when nontrivial right preconditioning is used.
%
%
\subsection{Optimal multilevel preconditioning}
\label{ssec:bpx}
The limitations highlighted above imply that a good preconditioner for quantum computing should result in a preconditioned matrix that is easy to construct directly and not as a product of the system matrix and some preconditioner. The BPX preconditioner \cite{BPX89} satisfies this requirement and is known to give a bounded condition number for the discrete Poisson-type problem \eqref{eq:weak-discrete} independent of the refinement level $L$. Its usual formulation is
\[ Pu = \sum_{\ell=1}^L \sum_{j=1}^{\dim \V_\ell} 2^{-\ell(2 - d)} (u, \Lambda_j^{(\ell)})_{L^2}\Lambda_j^{(\ell)}. \]
Here we assume a hierarchical decomposition of $\V_h$ based on the mesh hierarchy $\mathcal{G}_1, \dots, \mathcal{G}_L$, where $\mathcal{G}_\ell$ corresponds to the Cartesian grid of width $h_\ell=2^{-\ell}$ for $\ell=1,\ldots,L$ so that $\mathcal{G}_L = \mathcal{G}_h$ equals the finest mesh. On each mesh, we consider the Q1 finite element space~$\V_\ell=\operatorname{span}\{\Lambda_j^{(\ell)}\}_{j=1}^{\dim \V_\ell}$ spanned by the basis of hat functions $\Lambda_j^{(\ell)}$ associated with the interior vertices $x_j$ of $\mathcal{G}_\ell$. The symmetric preconditioning by $P$ leads to a bounded condition number that is independent of the number of levels $L$.
\begin{lemma}[Optimality of the BPX preconditioner] \label{lem:bpx}
Let $d \in \N, D \subset \R^d$ be a bounded domain and $\V_1 \subset \dots \subset \V_L \subset H^1_0(D)$ be finite element spaces. If for all $1 \le \ell \le L$ these spaces satisfy the Jackson estimates
\[ \inf_{v_\ell \in \V_\ell} \|u - v_\ell\|_{L^2} \le 2^{-\ell} C_1 \|u\|_{H^1}  \]
for all $u \in H^1_0(D)$ and the Bernstein estimates
\[ \|v_\ell\|_{H^1} \le 2^\ell C_2 \|v_\ell\|_{L^2} \]
for all $v_\ell \in \V_\ell$ with constants $C_1, C_2 > 0$ independent of $\ell$ and $L$, then the condition of the symmetrically preconditioned stiffness matrix
\[ \kappa(P^{1/2}SP^{1/2}) \in \bigO(1) \]
is bounded independently of $L$.
\end{lemma}
For a proof of this lemma we refer to \cite{Osw90, DK92, Xu92}.  Note that the hidden constant in the $\mathcal{O}$-notation may grow with the dimension $d$. The sufficient conditions for the lemma to hold are satisfied for our choice of uniformly refined grids and Q1 finite elements. The first assumption follows from the error estimate \eqref{eq:errorh}, and the second assumption is a consequence of the \textit{inverse inequality} for polynomials. 

In the classical context with the preconditioned conjugate gradient method, the explicit knowledge of $P^{1/2}$ is not required to apply symmetric preconditioning. No analogous argument is known for quantum computing. Since the non-symmetric left (or right) preconditioned system $PS$ generally has no bounded condition number, we actually have to split the preconditioner $P$ as is done for multilevel frames \cite{Osw94, HW89} used in the context of finite element methods \cite{Griebel:1994*4,Griebel:1994*2,HSS08}.
For the split one considers the corresponding generating system
\[\mathcal F \coloneqq \{f_{j,\ell} \coloneqq 2^{-\ell(2-d)/2}\Lambda_j^{(\ell)} \mid 1 \le \ell \le L, 1 \le j \le \dim \V_\ell\}.\]
Let $P \colon \R^{|\mathcal F|} \to \V_L$ be the linear function given by
\[ Fc = \sum_{\ell=1}^L\sum_{j=1}^{\dim \V_\ell} c_{j,\ell} f_{j,\ell}. \]
Then $P = FF^T$ can be written in terms of $F$ \cite[Section 4.2.1]{Osw94} and we obtain the preconditioned system
\begin{equation} \label{eq:preconditioned-system}
    F^TSF y = F^Tr, \qquad c = Fy.
\end{equation}
Note that $F$ is not equal to $P^{1/2}$, since it is not symmetric or even square. Nevertheless, Lemma \ref{lem:bpx}  serves as a bound on the condition number of \eqref{eq:preconditioned-system}, since the preconditioned matrix $F^TSF$ is spectrally equivalent to $P^{1/2}SP^{1/2}$ except for the nontrivial kernel, namely
\[\sigma(F^TSF) \setminus \sigma(P^{1/2}SP^{1/2}) = \{0\}.\]
Krylov methods in classical computers \cite{Griebel:1994*4,Griebel:1994*2}, as well as the quantum solvers based on quantum signal processing \cite{GSLW19} or adiabatic quantum computing \cite{AL22, LT20} that we intend to use, ignore the nontrivial kernel, so the effective condition of both systems is the same, and so is the effectiveness of the solvers.

\section{Quantum Realization of the Finite Element Method}
This section presents the main results of this paper, including a block encoding of the preconditioned stiffness matrix, optimal complexity bounds of the resulting quantum FEM, and an optimized quantum circuit realizing the method.

\subsection{Implementation of the preconditioned stiffness matrix} \label{sec:implementation}
This section presents a quantum implementation of the preconditioned  stiffness matrix $F^TSF$ from \eqref{eq:preconditioned-system}, which results from the finite element discretization of the partial differential equation~\eqref{eq:modelproblem} and BPX preconditioning. Since $F^TSF$ is not sparse in general, this is a nontrivial task where standard techniques are not applicable. Instead, we will make heavy use of Proposition~\ref{prop:block-encoding-ops} to decompose the preconditioned system into sums, products, and tensor products of smaller matrices, which in turn represent finite element gradients independent of the diffusion coefficient $A$ and grid transfer operators that encode a hierarchical grid structure. We give an explicit implementation of such matrices for the case of uniform Cartesian grids on the unit cube in Theorem~\ref{thm:main2}. However, the remainder of the proof given in Theorem~\ref{thm:main1} does not make any assumptions on the domain $D$.

In addition to the continuous finite element spaces $\V_\ell$, our construction relies on the auxiliary spaces of discontinuous Q1 elements $\Q_\ell \supset \V_\ell$ on $\mathcal{G}_\ell$. For these spaces of dimension~$2^{d(\ell+1)}$, $L^2$-orthonormal bases are composed of $2^d$ local basis functions $\psi_{j,1}^{(\ell)}, \dots, \psi_{j,2^d}^{(\ell)}$ per grid cell $G_j \in \mathcal{G}_\ell$, $1 \le j \le 2^{d\ell}$.
Our decomposition uses matrix representations $C_\ell$ of size $d2^{d(\ell+1)}\times (2^\ell-1)^d$ of the gradient $\nabla \colon \V_\ell \to \Q_\ell^d$ acting on finite element spaces. Furthermore, we need the $2^{d(L+1)}\times 2^{d(\ell+1)}$ transfer matrices $T_{\ell,L}$ representing the inclusion~$\Q_\ell \to \Q_L$. The vector-valued space $\Q_\ell^d = \Q_\ell \otimes \R^d$ has an orthonormal basis with basis vectors $\psi_{j,k}^{(l)} \otimes e_s$ where $e_s$ is the $s$-th canonical basis vector.

\begin{remark}[Qubit register ordering] \label{rem:specific-encoding}
We assume a special encoding of the spaces $\Q_\ell$ and $\V_\ell$ in qubit registers. Specifically,  the basis function $\Lambda_j^{(\ell)} \in \V_\ell$ should be represented (ignoring ancilla registers) by the integer of its index $\ket{j}$, the basis function $\psi_{j,k}^{(\ell)} \in \Q_\ell$ by $\ket{j}\ket{k}$ and $\psi_{j,k}^{(l)} \otimes e_s$ should be represented by $\ket{j}\ket{s}\ket{k}$. For some block encodings, we will need a different order of qubit registers, which corresponds to a permutation $\pi$ of rows and columns of the encoded matrix. The permutation $\pi$ is unitary, and a block encoding is given by a number of $\SWAP$ instructions that change the order of qubits correspondingly.
\end{remark}

The following theorem constructs an abstract block encoding of the preconditioned matrix $F^TSF$ if the matrices $C_\ell$ and $T_{\ell,L}$ are provided as block encodings $U_{C_\ell}$ and $U_{T_{\ell,L}}$ in the sense of Definition~\ref{d:block} and the diffusion coefficient $A$ is given in the format $U_A$ of Proposition~\ref{prop:block-encoding-diag}.

\begin{theorem}[Block encoding of preconditioned stiffness matrix]\label{thm:main1}
Let $L \in \N$ and $U_{C_\ell}$ and $U_{T_{\ell,L}}$ be block encodings of the finite element gradient $C_\ell$ and transfer matrices $T_{\ell,L}$  for $1 \le \ell \le L$ respectively. Let additionally $U_A$ be a quantum circuit that computes the value of the diffusion coefficient $A$. We can then construct a block encoding of the matrix~$F^TSF$ with subnormalization less than
\[\frac{\beta}{\alpha}\tilde\gamma(U_A)\sum_{\ell = 1}^L (\tilde \gamma(U_{C_\ell}) \tilde \gamma(U_{T_{\ell,L}}))^2\]
with $\alpha$ and $\beta$ defined in \eqref{eq:coer} and normalization
\[\gamma(U_A)\sum_{\ell=1}^L2^{-\ell(2-d)}(\gamma(U_{C_\ell})\gamma(U_{T_{\ell,L}}))^2.\]
\end{theorem}
\begin{proof}
Due to the choice of an orthonormal basis, the $L^2(D)$ inner product of two $\Q_L$ functions is equal to the Euclidean inner product of their coefficients in basis representations. Recall that we assume the diffusion coefficient $A$ to be constant on the mesh cells meaning under the chosen basis that it corresponds to a block diagonal matrix $D_A$ with the $d \times d$ blocks $(D_A)_{jj}=A\vert_{G_j}$ representing the diffusion coefficient in the $j$-th grid cell $G_j\in \mathcal{G}_L$. We thus obtain a decomposition
\begin{equation} \label{eq:decomposition}
S = C_L^T(D_A \otimes \Id_{2^d})C_L
\end{equation}
of the stiffness matrix. Using Proposition \ref{prop:block-encoding-diag}, the block diagonal matrix $D_A$ can be realized efficiently on a quantum computer. The decomposition \eqref{eq:decomposition} gives rise to a decomposition of the preconditioned system
\begin{equation} \label{eq:main1:factorization}
F^TSF = C_F^T (D_A \otimes \Id_{2^d})C_F
\end{equation}
with $C_F = C_LF$. The factors $C_L$ and $F$ of $C_F$ are not well conditioned, so we can not gain an efficient implementation of $C_F$ using just \eqref{eq:op-mutiply}. Still, we can decompose this matrix further by using the hierarchical structure induced by the preconditioner.
The blocks of $F$ correspond, up to a scalar factor, to the inclusions $\V_\ell \to \V_L$ for $1 \le \ell \le L$.  It follows that $C_F$  also has a $1 \times L$ block structure where the $\ell$-th block for $1 \le \ell \le L$ corresponds to the gradient on the level $\ell$, specifically the function $\nabla \colon \V_\ell \to \Q_L^d$, scaled by a factor of $2^{-\ell(2-d)/2}$. This factors through $\Q_\ell^d$, which can be expressed as the following commutative diagram.
\[
\begin{tikzcd}
	\V_\ell \rar[hook]{F\vert_{\V_\ell}} \dar[swap]{2^{-\ell(2-d)/2}C_\ell} & \V_L \dar{C_L} \\
	\Q_\ell^d \rar[hook]{\tilde T_{\ell,L}} & \Q_L^d
\end{tikzcd}
\]
The matrix for the inclusion of the vector-valued function spaces $\Q_\ell^d \hookrightarrow \Q_L^d$ is given by
\[\tilde T_{\ell,L} \coloneqq \pi^T(\Id_d \otimes T_{\ell,L}) \pi\]
where $\pi$ is the permutation of qubits $\ket{j}\ket{s}\ket{k} \mapsto \ket{s}\ket{j}\ket{k}$, see also Remark \ref{rem:specific-encoding}. Thus, we can write
\[ C_F = \begin{bmatrix} 2^{-(2-d)/2}\tilde T_{1, L} C_1 & \dots & 2^{-\ell(2-d)/2}\tilde T_{\ell, L} C_\ell & \dots & 2^{-L(2-d)/2}\tilde T_{L,L} C_L \end{bmatrix}.\]
We would now like to estimate the subnormalization of the construction. First, note that the mapping $\Q_\ell \hookrightarrow \Q_L$ preserves the $L^2(D)$ inner product since it is an inclusion of spaces. As under the chosen bases of both $\Q_\ell$ and $\Q_L$, the $L^2$ product is equal to the Euclidean inner product, we can deduce that the columns of $T_{\ell,L}$ are orthonormal. This means that the matrices are perfectly conditioned, i.e.
\[ \kappa(T_{\ell,L}) = \kappa(\tilde T_{\ell,L}) = 1. \]
As such we can use \eqref{eq:op-mutiply} and \eqref{eq:op-block-horizontal} of Proposition~\ref{prop:block-encoding-ops} to see that the block encoding of $C_F$ has subnormalization
\[\tilde \gamma(C_F) = \bigg(\sum_{\ell = 1}^L (\tilde \gamma(U_{C_\ell}) \tilde \gamma(U_{T_{\ell,L}}))^2\bigg)^{1/2}.\]
The statement for the normalization then follows from using \eqref{eq:op-mutiply} of Proposition \ref{prop:block-encoding-ops} again on the factorization \eqref{eq:main1:factorization}. For the bound of the subnormalization, we use
\[|F^TSF| = |C_F^T(D_A \otimes \Id_{2^d})C_F| = \big|\big(D_A^{1/2} \otimes \Id_{2^d}\big) C_F\big|^2 \ge \kappa\big(D_A^{1/2}\big)^{-2}\,\big|D_A^{1/2}\big|^2\,|C_F|^2. \]
Note that the symmetric decomposition of $D_A^{1/2}$ is only for theoretical purposes and is not needed for the quantum computation. The subnormalization can then be estimated as
\begin{align*}
\tilde \gamma(C_F^T(D_A \otimes \Id_{2^d})C_F) &= \frac{\gamma(C_F^T(D_A \otimes \Id_{2^d})C_F))}{|C_F^T(D_A \otimes \Id_{2^d})C_F|} \le \frac{\kappa(D_A)\gamma(U_A)\gamma(C_F)^2}{|D_A||C_F|^2}\\
&
= \kappa(D_A)\tilde \gamma(U_A)\tilde \gamma(C_F)^2 \le \frac{\beta}{\alpha}\tilde\gamma(U_A)\tilde \gamma(C_F)^2. \qedhere
\end{align*}
\end{proof}

We emphasize that the coefficient $A$ does not have to conform to any special structure, such as being tensorial in the spatial dimensions, in order for such an implementation to be possible. Instead, Proposition \ref{prop:block-encoding-diag} merely requires that $A$ as a function $\R^d \to \R^{d \times d}$ is computable efficiently.

Up until now, the treatment of the implementation was independent of the domain $D$; however, implementing the matrices $T_{\ell,L}$ and $C_\ell$ now strongly depends on the structure of the mesh. As such, we will now again consider the simplest case of a uniform grid on the $d$-dimensional cube $[0, 1]^d$. The same approach should work for more general meshes; however, see also the comment at the end of Section~\ref{sec:encoding} and Remark \ref{rem:general-domain}.

\begin{theorem}[Block encoding of preconditioned stiffness matrix on Cartesian grids] \label{thm:main2}
Let $L, d \in \N$, $D \coloneqq [0, 1]^d$ and $\mathcal{G}_L$ be the regular grid on $D$. We can then construct a block encoding of $F^TSF$ with subnormalization less than $\tfrac{\beta}{\alpha}\tilde\gamma(U_A)d(L + \tfrac14\pi^2)$ and normalization~$4\gamma(U_A) dL$.
\end{theorem}
\begin{proof}
For our construction, we will exploit the tensorial structure of the problem, namely that the spaces $\V_\ell$ and $\Q_\ell$ can be decomposed as
\begin{equation} \label{eq:main2:tensor-structure}
\V_\ell = \V_{\ell,\text{1D}}^{\otimes d} \qquad \text{and} \qquad \Q_\ell = \Q_{\ell,\text{1D}}^{\otimes d}
\end{equation}
where $\V_{\ell,\text{1D}}$ and $\Q_{\ell,\text{1D}}$ are the finite elements spaces of P1 and discontinuous P1 finite elements on the one-dimensional domain $[0, 1]$ with cells of width $2^{-\ell}$. In terms of the encoding of $\V_\ell$ and $\Q_\ell$ given in Remark \ref{rem:specific-encoding} this means that the register corresponding to the spatial index $1 \le j \le 2^{d\ell}$ is decomposed into registers storing the index in each of the spatial dimensions
\[\ket{j} = \ket{j_1}\dots\ket{j_d}.\]
Similarly, for a basis function $\psi_{j,k}^{(\ell)} \in \Q_\ell$, the index $1 \le k \le 2^d$ is decomposed into $d$ 1-qubit registers
\[\ket{k} = \ket{k_1}\dots\ket{k_d}.\]
To use Theorem \ref{thm:main1} we now have to choose an orthonormal basis of $\Q_\ell$, and we do so by choosing a basis on $\Q_{\ell,\text{1D}}$. Consider the functions $\psi_0(x) = 1$ and $\psi_1(x) = 2\sqrt{3}x - \sqrt{3}$ on the reference cell $[0, 1]$. One can check that these are indeed orthonormal. For $1 \le j \le 2^\ell$, let $\phi_j$ be the linear mapping from $[0, 1]$ to $[(j-1)2^{-\ell}, j2^{-\ell}]$. Then $(2^{\ell/2} \phi_j \psi_k)_{1\le j \le 2^\ell, k = 0,1}$ is an orthonormal basis of $\Q_{\ell,\text{1D}}$. This induces an orthonormal basis of $\Q_\ell$ through the relation \eqref{eq:main2:tensor-structure}.
With this choice of basis, we also introduce one-dimensional equivalents of relevant matrices. Specifically, for $1 \le \ell \le L$, we write  $T_{\ell,k,\text{1D}}$ for the $2^{\ell+1} \times 2^{k+1}$ matrix encoding the inclusion $\Q_{\ell,\text{1D}} \subset \Q_{k, \text{1D}}$ and extend this definition to $\ell \le k \le L$. In addition, we define $C_{\ell,\text{1D}}$ and $R_{\ell,\text{1D}}$ as the $2^{\ell+1} \times 2^\ell - 1$ matrices encoding the derivative~$\partial_1 \colon \V_{\ell,\text{1D}} \to \Q_{\ell,\text{1D}}$ and the inclusion $\V_{\ell,\text{1D}} \subset \Q_{\ell,\text{1D}}$.
To match our assumption in Remark \ref{rem:specific-encoding} on the ordering of the basis elements, the induced basis of $\Q_\ell$ needs a slight reordering by permutation
\[ \pi_\ell \colon \ket{j_1}\dots\ket{j_d}\ket{k_1}\dots\ket{k_d} \mapsto \ket{j_1}\ket{k_1}\dots\ket{j_d}\ket{k_d} \]
which groups the indices belonging to one dimension together.

In addition to the tensor product structure of the function spaces, we also obtain a tensor structure of the gradient since
\[ (\nabla v)_s = \partial_s v = \underbrace{\Id \otimes \dots \otimes \Id}_{s-1} \otimes \partial_1 \otimes \underbrace{\Id \otimes \dots \otimes \Id}_{d - s} \]
for $1 \le s \le d$. On the matrix level, this means that $\pi_\ell C_\ell$ is a $d \times 1$ block matrix, where the $s$-th block is given by
\[ \underbrace{R_{\ell,\text{1D}} \otimes \dots \otimes R_{\ell,\text{1D}}}_{s-1} \otimes C_{\ell,\text{1D}} \otimes \underbrace{R_{\ell,\text{1D}} \otimes \dots \otimes R_{\ell,\text{1D}}}_{d - s}. \]
Additionally, the matrices $R_{\ell,\text{1D}}$ and $C_{\ell,\text{1D}}$ can be given explicitly as
\begin{align*}
    R_{\ell,\text{1D}} &=
2^{-\ell/2} \left(\Id_{2^\ell} \otimes\begin{bmatrix} \tfrac12 & \tfrac12 \\ \tfrac1{2\sqrt{3}} & -\tfrac1{2\sqrt{3}} \end{bmatrix}\right)
\begin{bmatrix} I_\ell \\ N_\ell \end{bmatrix} \\
C_{\ell,\text{1D}} &= 2^{\ell/2} \left(\Id_{2^\ell} \otimes\begin{bmatrix} 1 & -1 \\ 0 & 0 \end{bmatrix}\right)
\begin{bmatrix} I_\ell \\ N_\ell \end{bmatrix}
\end{align*}
where $I_\ell = \Id_{2^\ell,2^\ell-1}$ and $N_\ell \in \R^{2^\ell \times 2^\ell-1}$ is the matrix of the mapping $\ket{k} \mapsto \ket{k+1}$,
\[
I_\ell \coloneqq \begin{bmatrix}
1 & \dots & 0 \\
0 &\ddots & \vdots \\
\vdots & & 1 \\
0 & \dots & 0
\end{bmatrix}
\qquad \text{and} \qquad
N_\ell \coloneqq \begin{bmatrix}
0 & \dots & 0 \\
1 & & \vdots \\
\vdots &\ddots & 0 \\
0 &\dots & 1
\end{bmatrix}.
\]
Block encodings of both $N_\ell$ and $I_\ell$ can be implemented, even without the use of extra ancilla qubits. See for example \cite{Dra00} for an implementation of integer addition.
We can easily check that the encoding for $R_{\ell,\text{1D}}$ has a normalization of $2^{-\ell/2}$ while the encoding for $C_{\ell,\text{1D}}$ has a normalization of $2^{\ell/2+1}$ leading by \eqref{eq:op-block-horizontal} to $2^{-\ell(2-d)/2} \gamma(U_{C_\ell}) = 2\sqrt{d}$. To obtain a good bound on the subnormalization of these matrices, we do not use Proposition~\ref{prop:block-encoding-ops} but instead compute the subnormalization by its definition. Specifically, since $C_{\ell,\text{1D}}^TC_{\ell,\text{1D}}$ is a tridiagonal Toeplitz matrix, we can compute its largest eigenvalue \cite{Yueh05} as
\[ |C_{\ell,\text{1D}}^TC_{\ell,\text{1D}}| = 2^{\ell+2} - 2^{\ell+2}\sin^2\big(2^{-(\ell+1)}\pi\big) \]
and by $\sin(x) \le x$ we get
\[
|C_{\ell,\text{1D}}^TC_{\ell,\text{1D}}| \ge 2^{\ell+2} - 2^{-\ell}\pi^2
\]
so the subnormalization of this block encoding is bounded by $\sqrt{1 + \tfrac14 2^{-2\ell}\pi^2}$. Similarly, considering the eigenvalues of $R_{\ell,\text{1D}}^TR_{\ell,\text{1D}}$ we get
\[
|R_{\ell,\text{1D}}^TR_{\ell,\text{1D}}| \ge 2^{-\ell} - \tfrac{1}{3} 2^{-3\ell}\pi^2,
\]
which means that the subnormalization of the block encodings is bounded by $\sqrt{1 + \tfrac{1} {3} 2^{-2\ell}\pi^2}$. By \eqref{eq:op-block-horizontal} we can then see that
\[
\tilde \gamma(U_{C_\ell}) \le \sqrt{d\big(1 + \tfrac13 2^{-2\ell}\pi^2\big)}. 
\]

Finally, the matrices $T_{\ell,L}$ can also be factorized. Again, due to \eqref{eq:main2:tensor-structure}, we get the property $\pi_L T_{\ell,L} \pi_\ell^T = T_{\ell,L,\text{1D}}^{\otimes d}$. This can be further decomposed using $\Q_{\ell, \text{1D}} \hookrightarrow \Q_{\ell+1, \text{1D}} \hookrightarrow \dots \hookrightarrow \Q_{L, \text{1D}}$ into 
\[\pi_LT_{\ell,L}\pi_\ell^T = (T_{L-1,L, \text{1D}} \dots T_{\ell,\ell+1, \text{1D}})^{\otimes d}.\]
The matrices $T_{\ell,\ell+1,\text{1D}}$ are explicitly given by
\[ T_{\ell,\ell+1, \text{1D}} = \Id_{2^\ell} \otimes \frac{1}{\sqrt{2}}\begin{bmatrix}
1&-\sqrt{3}/2 \\
0&1/2 \\
1&\sqrt{3}/2 \\
0&1/2
\end{bmatrix}.\]
They have orthonormal columns and thus a condition number of $1$. This means $T_{\ell,L}$ has condition number $1$ and the block encoding has normalization and subnormalization $1$.

By Theorem \ref{thm:main1} the normalization for the block encoding of the preconditioned stiffness matrix is then $4\gamma(U_A) dL$. The subnormalization is computed as
\begin{align*}
\frac{\beta}{\alpha}\tilde\gamma(U_A)\sum_{\ell=1}^L \tilde \gamma(U_{C_\ell})^2
&= \frac{\beta d\tilde\gamma(U_A)}{\alpha} \bigg(L + \frac{\pi^2}{3}\sum_{\ell=1}^L 2^{-2\ell} \bigg)
\le \frac{\beta d\tilde\gamma(U_A)}{\alpha}\Big(L + \frac{\pi^2}{4}\Big)
\end{align*}
where the sum was bounded by a geometric series.
\end{proof}

\begin{remark} \label{rem:general-domain}
Our technique can be easily adapted to domains $D \subset [0, 1]^d$ such that $D$ is a union of cells at some level $\ell_0$. This is because the matrices $T_{\ell,L}$ and $C_\ell$ for $\ell_0 \le \ell \le L$ are obtained from the ones defined in Theorem \ref{thm:main2} by a simple projection. We would need to truncate the hierarchical sequence of meshes at $\ell_0$, meaning that the condition number would not be independent of $\ell_0$.
\end{remark}

\subsection{Complexity of quantum FEM}

We will now comment on the complete algorithm and measurement procedure.  As stated earlier, because of the symmetric preconditioning, it is not possible to measure quantities of interest of the form~$\bra{c}\!M\!\ket{c}$, but rather the goal should be computing a linear functional $m^Tc$ of the solution.

\begin{theorem}[Complexity of quantum FEM] \label{thm:main3}
Let $\tol > 0$, $d \in \N$, and $D = [0, 1]^d$.
Let $U_r$ and $U_m$ be quantum circuits that construct the normalized states corresponding to the preconditioned right-hand side and the quantity of interest, $\tilde r \coloneqq F^T r$ and $\tilde m \coloneqq F^T m$, as outlined in Appendix~\ref{app:implementation-vectors}. Furthermore, let $U_A$ be a quantum circuit that computes the value of the coefficient as defined in Definition~\ref{prop:block-encoding-diag}. The quantity of interest can then be computed as the amplitudes of a qubit up to precision $\tol$ using $\bigO\big(\poly\log\tol^{-1}\big)$ uses of $U_r$, $U_m$, $\bigO(1)$ uses of $U_A$, and $\bigO\big(\poly \log \tol^{-1}\big)$ additional gates. To measure the amplitude, the algorithm is repeated $\bigO\big(\tol^{-1}\big)$ times, leading to a total runtime of $\bigO\big(\tol^{-1}\poly\log\tol^{-1}\big)$.
\end{theorem}
\begin{proof}
We decompose the linear system as described in Theorems \ref{thm:main1} and \ref{thm:main2}, specifically in the form of
\begin{equation} \label{eq:main3:decomposition}
    F^TSF = C_F^T(D_A \otimes \Id_{2^d})C_F
\end{equation}
by which we also get block encodings for all factors.
As mentioned in Section~\ref{ssec:bpx}, while the solution to the system
\[
F^TSFy = F^Tr = \tilde r
\]
is unique only up to the kernel of $F$, we still obtain $c = Fy$ for any such solution~$y$. We take $y = (F^TSF)^+\tilde r$ and write
\[ m^Tc = m^TS^{-1}r = m^TF (F^TSF)^+ \tilde r = \tilde m^T (F^TSF)^+ \tilde r. \]
Inserting the decomposition \eqref{eq:main3:decomposition} of the preconditioned matrix gives
\[m^Tc = {\tilde m}^T(C_F^T(D_A \otimes \Id_{2^d})C_F)^+\tilde r. \]
A block encoding for this pseudoinverse can be constructed using the previous theorems, but if an encoding of $D_A^{1/2}$ is available, a faster algorithm is obtained by decomposing this system further.
We can use the fact that $(Y^T)^+Y^+ = (YY^T)^+$ for any matrix $Y$, c.f.~\cite{Gre66}, to see
\[m^Tc = (((D_A^{1/2} \otimes \Id_{2^d})C_F)^+\tilde m)^T((D_A^{1/2} \otimes \Id_{2^d})C_F)^+\tilde r.  \]
By Lemma~\ref{lem:bpx} we know $\kappa((D_A^{1/2} \otimes \Id_{2^d})C_F) = \sqrt{\kappa(F^TSF)} = \bigO(1)$ and using Theorems~\ref{thm:main1} and~\ref{thm:main2} we can construct a block encoding where the subnormalization is similarly bounded by $\bigO(\log\tol^{-1})$. As such we can first compute the least square solutions $((D_A^{1/2} \otimes \Id_{2^d})C_F)^+\tilde m$ and $((D_A^{1/2} \otimes \Id_{2^d})C_F)^+\tilde r$ in parallel, which takes $\bigO\big(\poly\log\tol^{-1}\big)$
uses of either oracles $U_r$ and $U_m$ and a proportional number of additional operations following Theorem \ref{thm:qlsa}. The result is a state of the form
\[
(((D_A^{1/2} \otimes \Id_{2^d})C_F)^+ \otimes \Id_2)
\begin{bmatrix}
    \tilde m \\
    \tilde r
\end{bmatrix}
=
\begin{bmatrix}
    ((D_A^{1/2} \otimes \Id_{2^d})C_F)^+\tilde m \\
    ((D_A^{1/2} \otimes \Id_{2^d})C_F)^+\tilde r
\end{bmatrix}
\]
The quantity of interest is computed by measuring with the observable $\Id \otimes X$. This procedure gives a roughly quadratic improvement in runtime with respect to condition~$\kappa(F^TSF)$ over simply computing the solution $\ket{c}$ and measuring using a swap test. The approach is similar to the ``Modified Hadamard test'' described in~\cite{Luo20}. This can be combined with Monte Carlo sampling to give a runtime roughly square in $\tol^{-1}$, for which the complete circuit can be seen in Figure \ref{fig:complete-algorithm}.
Alternatively, amplitude estimation \cite{SUR+20} only requires $\bigO(\tol^{-1})$ runs of the Hadamard test, giving the optimal runtime of $\bigO(\tol^{-1}\poly\log \tol^{-1})$.
\end{proof}

\begin{figure}
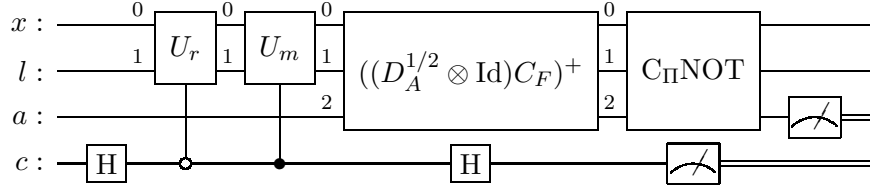

\begin{center}
\includestandalone{circuits/QC}
\end{center}
\caption{ \label{fig:complete-algorithm}
High-level overview of an efficient algorithm to compute the quantity of interest $m^TS^{-1}r$, which can be computed as $P(a = 0, c = 0) - P(a = 0, c = 1)$. The gate labeled $((D_A^{1/2} \otimes \Id_{2^d})C_F)^+$ represents the linear solver.
}
\end{figure}

\begin{remark}[Dependence on $d$]
Note that the condition number $\kappa(F^TSF)$ of the preconditioned finite element system may depend in a critical way on the dimension~$d$ \cite[Remark 3.2.]{Bac23}. To the best of our knowledge, the asymptotic behavior of the BPX preconditioner has not been analyzed with respect to its dependence on $d$ as $d \to \infty$. The best bound for the unpreconditioned system known to us is exponential in $d$, but we do not see this behavior in numerical experiments with low dimension. If instead the condition depends only polynomially on the dimension, then our scheme would give an exponential speed-up with respect to $d$ over classical methods, specifically an asymptotic runtime of $\bigO(\tol^{-1}\poly(d \log \tol^{-1}))$.
\end{remark}

\begin{remark}[Advantage over classical computers]
    The above theorem assumes that $\tol \approx h$, which corresponds to $H^2(D)$ regularity of the solution. In cases of globally low regularity beyond isolated singularities, a more restrictive mesh scaling, $h\approx\tol^{1/s}$ with $s>0$, may be necessary. This means that classical methods require $\bigO(h^{-d/s})$ operations to achieve a solution with the desired accuracy, making a quantum advantage possible as soon as $d > s$. The potential speedup scales as $\tol^{-(1-d/s)}$, becoming more pronounced when $s \ll 1$.
    Our method consistently realizes this potential speedup since the mesh width~$h$ only affects runtime logarithmically. This is the first work to rigorously establish such advantage for the case $s = 1$. Previous methods that achieved runtimes close to $\bigO(\tol^{-1})$, such as \cite{CLO21}, relied on significantly stronger regularity assumptions.
    For higher-regularity problems where $s > 1$ is feasible, the quantum advantage naturally diminishes. 
\end{remark}

\subsection{Optimized quantum circuit} \label{ssec:optimized}

We now turn to our actual implementation of the block encoding for $C_F$. Indeed, this could be done using the operations of Proposition~\ref{prop:block-encoding-ops} and their concrete implementations in Appendix \ref{app:proof}, but this will lead to a large overhead which can mostly be attributed to two reasons. First, the inner projection in the multiplication (see Figure \ref{fig:ops-circuits}) can be avoided in some cases, and with our implementation none of these projections are needed at all. Second, there is often a more efficient way to implement the concatenation of matrices than to use controlled applications of the blocks. Consider, for example, the block matrix
\[ \begin{bmatrix} I_L & N_L & I_{L-1} & \dots & N_1 \end{bmatrix}^T \]
which arises from the construction of Theorem \ref{thm:main2}. The $2\ell$-th block for $1 \le \ell \le L$ is simply encoded by identity, while the $(2\ell+1)$-th block is implemented by increasing the integer encoded in the $(L-\ell)$ highest bits of the work register or equivalently by first shifting the bits of the work register right by the $\ell$ bits, then increasing the complete work register, and then shifting the work register left again by $\ell$ bits. With the latter procedure, all incrementing operations overlap and only one such subcircuit is needed. A similar technique is used to compress the circuit that implements the matrices $T_{\ell,L}$. An overview of the complete quantum circuit for $U_{C_F}$ is given in Figure \ref{fig:complete-block-encoding}. The gates denoted with $R$ are related to the implementation of $R_{\ell,\text{1D}}$ and $C_{\ell,\text{1D}}$, while those denoted with $T$ are related to the implementation of $T_{\ell,\ell+1,\text{1D}}$ from Theorem \ref{thm:main2}.

To solve the linear system and estimate a quantity of interest $m^TS^{-1}r$ we restrict ourselves to the constant coefficient case $D_A = \Id$ and first compute the vectors $C_F^+F^Tr$ and $C_F^+F^Tm$ as explained in Theorem \ref{thm:main3}.
For this we use a quantum linear solver based on the Quantum Singular Value Transformation (QSVT) \cite{GSLW19} to compute the pseudoinverse~$C_F^+$. Although the runtime of this method scales quadratically in $\keff$ it is very efficient in the well-conditioned regime which we consider here.
The QSVT works by transforming the initial state $\ket{\tilde r} \coloneqq F^Tr$ by a polynomial $p$ of the input matrix $C_F$, in other words, one can produce the state $p(C_F)\ket{\tilde r}$. This polynomial must satisfy $|p(z)| \le 1$ for $z \in [-1, 1]$. We aim to approximate the function $g(z) = 1/z$ on the domain $[-1, -\kappa] \cup [\kappa, 1]$ using $p$, since $g(C_F) = C_F^+$ is the pseudoinverse of $C_F$. Following the construction in \cite{GSLW19}, an approximation is realized by the polynomial 
\[ \tilde p(z) = 4\sum_{j = 0}^J (-1)^j\frac{\sum_{k=j+1}^K \binom{2K}{K+k}}{2^{2K}}t_{2j+1}(z),\]
where $t_{2j+1}$ is the Chebyshev polynomial of degree $2j+1$, $j=0,\ldots,J$, and the parameters~$J,K\in\N$ can be chosen appropriately given $\kappa$ and the error tolerance $\tol$. Specifically, \cite{GSLW19} shows the choices 
$$K = \big\lceil \kappa^2 \log(\kappa\tol^{-1}) \big\rceil \quad\text{and}\quad J = \bigg\lceil \sqrt{K \log(4K\tol^{-1})} \bigg\rceil$$
lead to sufficient error bounds. The polynomial $\tilde p$ is then normalized to not exceed the range~$[-1, 1]$, resulting in the final polynomial $p$. For the purpose of the QSVT, the normalized polynomial $p$ is translated into so-called \textit{phase angles}, which we compute using the software \texttt{pyqsp} \cite{MRTC21}.
For technical details, we refer to our complete implementation in the Qiskit framework \cite{Qis23} available at \url{https://github.com/MDeiml/quantum-bpx}.

\begin{figure}
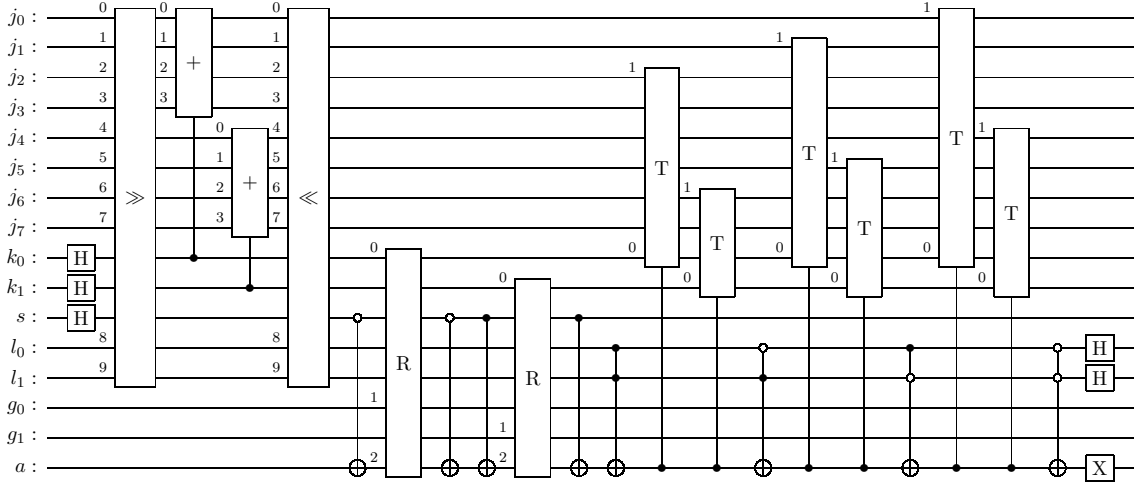

    \hspace{-5ex}
    \includestandalone{circuits/C_P}
    \caption{ \label{fig:complete-block-encoding}
    Block encoding (without $\CNOT{\Pi}$) of $C_F$ for $d = 2$ and $L = 4$}
\end{figure}

%
%
\section{Numerical Experiments}\label{sec:numerics}
In this section, we analyze the proposed method through numerical experiments. First, we execute the quantum circuits on both real and simulated hardware to estimate the impact of noise present in current noisy intermediate-scale quantum (NISQ) hardware. We then compare the performance of our method with preconditioning to that of a method that does not use preconditioning to demonstrate its effectiveness. Both experiments focus on solving the model problem \eqref{eq:modelproblem} with constant coefficients, where $A \equiv 1$ and $f \equiv 1$. We aim to approximate
\begin{equation} \label{eq:quantity-of-interest}
\int_D u \dx = \frac{1}{12}
\end{equation}
which is achieved by setting $m = r = 2^{-L}(1,\dots,1)^T$.

\subsection{Realization on NISQ hardware}
Here, we analyze the dependence of our algorithm on realistic noise, which is unavoidable in current and near-term quantum computers, and address the viability of the practical use of our method.

After validating the correctness of the circuit of Section~\ref{ssec:optimized} using a noiseless simulator, we attempt to compute \eqref{eq:quantity-of-interest} with $L = 4$ on the \texttt{ibm\_sherbrooke} quantum computer, which has a two-qubit gate error of roughly $\epsilon_2 \approx 7 \cdot 10^{-3}$. We choose $K = 27$ for the QSVT solver by the above formula with $\kappa = \kappa(C_F) \approx 2.8$ and $\tol^{-1} = 0.1$, but we vary the number of steps $J$. The results of this experiment can be seen in Figure \ref{fig:noise-example}. Each data point contains $10\,000$ measurement samples. Sadly, the noise seems to dominate, giving a relative error of about $25\%$.

We also report more recent results from the same experiment for $J=2$ performed on the \texttt{ibm\_marrakesh} quantum computer, which has a two-qubit gate error of approximately $\epsilon_2 \approx 3 \times 10^{-3}$. These data were collected for \cite{DP25_mfo} while this manuscript was under review. The reduction in hardware noise translates directly into an improved average relative error of about $13\%$ for this very small sample size.

These results are not surprising, as can be shown by giving a rough estimate of the expected noise. The transpiled circuits are quite compact, with the block encoding of $C_F$ including both $\CNOT{\Pi}$ operations requiring only $13$ qubits and $300$ two-qubit gates for $d = 1$ and $L = 4$. Still, the complete circuit using the QSVT based solver described in Section~\ref{ssec:optimized} requires about $300(2J+1)$ two-qubit gates (ignoring the gates needed to encode the right-hand side).
From these observations, one can estimate the expected error on a quantum computer with realistic noise, showing that a two-qubit gate error of
\[\epsilon_2 \lesssim \frac{0.1}{300 (2J+1)} \approx \frac{3.3 \cdot 10^{-4}}{2J + 1}\]
is needed to reliably obtain an absolute noise level below $0.1$. This is still beyond the capability of the quantum computer used. Even cutting-edge quantum computers only achieve a two-qubit error rate of $\epsilon_2 \approx 1.8 \cdot 10^{-3}$ \cite{MBA+23}, meaning that a further hardware improvement of about one order of magnitude is needed to make our approach practical. This prediction does not take into account the various error mitigation and correction methods that already exist \cite{DMN13,CBB+23} or any future improvements to such methods.

To back up this assessment with a numerical experiment, we simulate our algorithm to compute \eqref{eq:quantity-of-interest} with $L = 4$ without hardware noise, as well as with a two-qubit error ranging from $7 \cdot 10^{-3}$ to $10^{-3}$. The noise for one-qubit gates is scaled $10^{-2}$ times the two-qubit error. For each data point, we compute $200$ runs to estimate a 95\% confidence interval of the predicted quantum solution. Each run contains $10\,000$ measurement samples.
We found that in the presence of noise, the accuracy of the estimate can be improved by deviating from the procedure described in Theorem \ref{thm:main3}. Specifically, by projective measurement of the solution state we obtain $C_F^+F^Tr / |C_F^+F^Tr|$, then apply the block encoding of $C_F$ again, giving a state with norm $1 / |C_F^+F^Tr|$. This can be used to estimate $|C_F^+F^Tr|$ and analogously $|C_F^+F^Tm|$. By computing the inner product of the normalized vectors, we can approximate the quantity of interest. We believe that this improved resilience to noise is due to the projective measurement filtering out a lot of runs where larger noise was introduced. The final application of $C_F$ then introduces fewer additional errors than the computation of the pseudoinverse.
The results can be seen in Figure \ref{fig:noise-example}. For low noise, all data points are within a reasonable distance of the true solution, and the QSVT error is lower for $J = 3,4$ than it is for $J = 2$. However, note that this comes with an increasing error due to noise, with $J = 4$ converging the slowest.
These results nicely show that one could obtain an estimate of the quantity of interest within a relative error of about $0.1$ given a quantum computer with a two-qubit error of $10^{-3}$, which even slightly exceeds our prediction above. 
Although a trivial undertaking on a classical computer, this demonstration of practical feasibility of the method on current hardware represents a significant milestone. Until recently, quantum hardware was unable to process the circuit depth required for even the simplest finite element problem in one dimension. Our results indicate that only minor hardware improvements are required to solve prototypical PDEs with reasonable accuracy.
\begin{figure}
	\begin{tikzpicture}
		\begin{axis}[%
		width=0.88\textwidth,
		height=0.35\textwidth,
		at={(0\textwidth,0\textwidth)},
		scale only axis,
		unbounded coords=jump,
		xlabel style={font=\color{white!15!black}},
		xlabel={two-qubit gate error},
        xmin = 0.00002,
        xmax = 0.01,
        xmode=log,
        ymax = 0.098,
        xtick = {0.00003, 0.0001, 0.001, 0.003, 0.007},
        xticklabels = {noise-free, $10^{-4}$,$10^{-3}$,$3 \cdot 10^{-3}$, $7 \cdot 10^{-3}$},
        xminorticks = false,
        ytick = {0.05, 0.06, 0.07, 0.08, 0.09},
        yticklabels = {$0.05$, , , , $0.09$},
        extra y ticks = {0.08333333333},
        extra y tick labels = {QoI},
        extra y tick style={grid=major},
        scaled ticks=false,
		ylabel style={font=\color{white!15!black}},
		ylabel={result},
		axis background/.style={fill=white},
		legend columns = 4,
		legend style={
			at={(0.98,0.93)},
			anchor=east,
			legend cell align=left,
			align=left,
			draw=white!15!black,
			font=\Small,
            /tikz/every even column/.append style={column sep=1ex}
		},
		]
		\addplot [mark=o, color=mycolor2, only marks, very thick, mark size=2pt, error bars/.cd,
            y dir = both, y explicit, error bar style={line width=0.8pt},
            error mark options={
                rotate=90,
                mark size=4pt,
                line width=0.8pt
            }]
        table [col sep=semicolon, x expr = {\thisrowno{0} / 10^0.05}, y expr = {\thisrowno{1}}, y error plus expr = {(\thisrowno{4} - \thisrowno{1})}, y error minus expr = {(\thisrowno{1} - \thisrowno{5})}] {data/noise_error_3.csv};
		\addplot [mark=square, color=mycolor1, only marks, very thick, mark size=2pt, error bars/.cd,
            y dir = both, y explicit, error bar style={line width=0.8pt},
            error mark options={
                rotate=90,
                mark size=4pt,
                line width=0.8pt
            }]
        table [col sep=semicolon, x expr = {\thisrowno{0}}, y expr = {\thisrowno{1}}, y error plus expr = {(\thisrowno{4} - \thisrowno{1})}, y error minus expr = {(\thisrowno{1} - \thisrowno{5})}] {data/noise_error_4.csv};
  	\addplot [mark=triangle, color=mycolor3, only marks, very thick, mark size=3pt, error bars/.cd,
            y dir = both, y explicit, error bar style={line width=0.8pt},
            error mark options={
                rotate=90,
                mark size=4pt,
                line width=0.8pt
            }]
        table [col sep=semicolon, x expr = {\thisrowno{0} * 10^0.05}, y expr = {\thisrowno{1}}, y error plus expr = {(\thisrowno{4} - \thisrowno{1})}, y error minus expr = {(\thisrowno{1} - \thisrowno{5})}] {data/noise_error_5.csv};
        \addlegendentry{$J = 2$}
        \addlegendentry{$J = 3$}
	      \addlegendentry{$J = 4$}
        
		\addplot [mark=x, color=mycolor2, only marks, mark size=4pt]
        coordinates {(0.007 / 10^0.05, 0.05046828190356076) (0.007 / 10^0.05, 0.05926372798682381) (0.007 / 10^0.05, 0.05769232935536537)};
		\addplot [mark=x, color=mycolor1, only marks, mark size=4pt]
        coordinates {(0.007, 0.0547182710082045) (0.007, 0.07356420041686254) (0.007, 0.054477300723894596)};
		\addplot [mark=x, color=mycolor3, only marks, mark size=4pt]
        coordinates {(0.007 * 10^0.05, 0.05439757784478645) (0.007 * 10^0.05, 0.06915382060608485) (0.007 * 10^0.05, 0.04778295544160281)};
		\addplot [mark=x, color=mycolor2, only marks, mark size=4pt]
        coordinates {(0.003 / 10^0.05, 0.0603373126010106) (0.003 / 10^0.05, 0.06218426626148073) (0.003 / 10^0.05, 0.08300781250000001) (0.003 / 10^0.05, 0.06586251213372406)};
		\end{axis}
        \draw (0,0) node at (1.9, 0.035) {//};
	\end{tikzpicture}%
	\caption{\label{fig:noise-example}
The results for hardware execution and noisy simulation of the algorithm with $J=2,3,4$. The gray line corresponds to the true value of the quantity of interest (QoI) of the PDE solution. Cross markers ($\times$) indicate the outcomes of runs on two different instances of real hardware. Other markers are the empirical expectations of the simulated runs under varying two-qubit gate error in the noise model, while the bars indicate a confidence interval of $95\%$.
	}
\end{figure}

\subsection{Numerical validation of optimal convergence rate}
To evaluate the effectiveness of our preconditioning approach, we compare it with an equivalent method that solves the original linear system $Sc = r$ without preconditioning. This comparison is made by directly computing the effect of the solver, specifically the matrix polynomials $p(C_F)$ and $p(D_A^{1/2}C_L)$—rather than using simulated quantum computers, since the latter limit the maximum discretization due to large memory requirements.

As a model problem, we solve \eqref{eq:modelproblem} in one and two dimensions with the parameters previously specified. Both the preconditioned and non-preconditioned methods are applied to linear systems with increasing discretization levels $L$, ranging from $3$ to $14$ for $d=1$ and to $8$ for $d=2$. The linear solver tolerance is set to $2^{-L}$, and the parameter $\keff$ is optimized to be as low as possible while maintaining this tolerance. The results, shown in Figure \ref{fig:condition}, confirm that the dependence of our method on the inverse error is logarithmic, that is, $J \sim \log \tol^{-1}$, while the non-preconditioned method shows a linear dependence, $J \sim \tol^{-1}$. The results for both $d = 1$ and $d = 2$ show a similar behavior in this regard. It is important to note that this scaling is with respect to the square root of the actual condition number, $\kappa(S) \sim \tol^{-2}$, due to the symmetric factorization outlined in Theorem \ref{thm:main3}, which is applied to both preconditioned and non-preconditioned systems.

In general, one cannot expect quantum linear solvers to outperform a linear dependence on the effective condition number, even for linear systems of relatively small dimension~$N$. This contrasts with classical non-stationary methods, such as the preconditioned conjugate gradient method, which is guaranteed to converge to the exact solution within $N$ steps. This is achievable because the matrix polynomial computed by the CG method is continuously adapted to the system matrix. In contrast, current quantum linear solvers more closely resemble stationary methods, where a fixed polynomial is used for all system matrices \cite{BBC+94}.

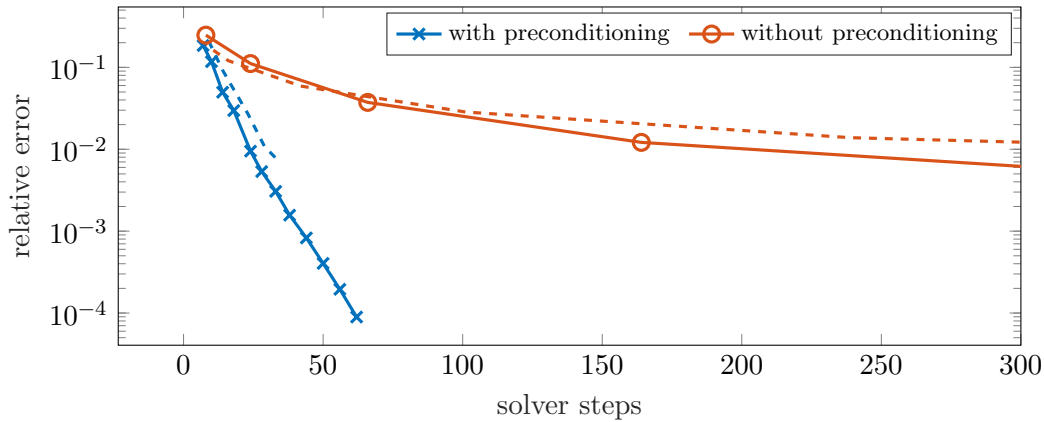
\begin{figure}
	\begin{tikzpicture}
		\begin{axis}[%
		width=0.8\textwidth,
		height=0.3\textwidth,
		at={(0\textwidth,0\textwidth)},
		scale only axis,
		unbounded coords=jump,
        xmax=300,
        ymode=log,
		xlabel style={font=\color{white!15!black}},
		xlabel={solver steps},
		ylabel style={font=\color{white!15!black}},
		ylabel={relative error},
		axis background/.style={fill=white},
		legend columns = 4,
		legend style={
			at={(0.99,0.92)},
			anchor=east,
			legend cell align=left,
			align=left,
			draw=white!15!black,
			font=\Small,
            /tikz/every even column/.append style={column sep=1ex}
		},
		]
		\addplot [mark=x, color=mycolor1, very thick, mark size=3pt]
        table [col sep=semicolon, x index = {2}, y index = {3}] {data/condition.csv};
        \addlegendentry{with preconditioning}
		\addplot [mark=o, color=mycolor2, very thick, mark size=3pt]
        table [col sep=semicolon, x index = {0}, y index = {1}, skip coords between index={8}{12}] {data/condition.csv};
        \addlegendentry{without preconditioning}
		\addplot [dashed, color=mycolor1, very thick, mark size=3pt]
        table [col sep=semicolon, x index = {6}, y index = {7}, skip coords between index={6}{12}] {data/condition.csv};
		\addplot [dashed, color=mycolor2, very thick, mark size=3pt]
        table [col sep=semicolon, x index = {4}, y index = {5}, skip coords between index={6}{12}] {data/condition.csv};
  
		\end{axis}
	\end{tikzpicture}%
	\caption{\label{fig:condition}
Comparison in needed solver steps of our method with a non-preconditioned approach for $d = 1$ (solid line) and $d = 2$ (dashed line). The graph includes the error from discretization and the quantum solver, but not errors due to hardware noise. The dependence on the error of our preconditioned approach is logarithmic, while the dependence for the non-preconditioned method is linear.
	}
\end{figure}

%
%
\section{Conclusion}
In this research article, we have demonstrated the feasibility of the finite element method for quantum computing, achieving a significant speed-up by optimizing the runtime of our proposed algorithm through the integration of the well-established BPX preconditioner. The method not only realizes optimal performance in terms of error tolerance, a prospect previously only theorized \cite{MP16}, but is also characterized by its low quantum circuit depth, bringing it excitingly close to the abilities of real quantum computers.

However, one aspect that remains to be fully explored is how the efficiency of the algorithm is affected by high dimensions. While this question will hopefully be answered empirically on real quantum computers in the not too distant future, its mathematical understanding requires novel and refined techniques of numerical analysis to obtain dimension-robust condition number estimates of preconditioned finite element discretizations, which have so far been explored up to dimension three or accepting exponentially growing dimension dependencies. Therefore, future work should aim to clarify the pre-asymptotic behavior of our algorithm in high dimensions.

%
%
\bibliographystyle{amsplain}
\bibliography{references}

\providecommand{\bysame}{\leavevmode\hbox to3em{\hrulefill}\thinspace}
\providecommand{\MR}{\relax\ifhmode\unskip\space\fi MR }
\providecommand{\MRhref}[2]{%
  \href{http://www.ams.org/mathscinet-getitem?mr=#1}{#2}
}
\providecommand{\href}[2]{#2}
\begin{thebibliography}{10}

\bibitem{AJL06}
D.~Aharonov, V.~Jones, and Z.~Landau, \emph{A polynomial quantum algorithm for
  approximating the jones polynomial}, Proceedings of the thirty-eighth annual
  ACM symposium on Theory of computing, 2006, pp.~427--436.

\bibitem{Amb10}
A.~Ambainis, \emph{Variable time amplitude amplification and a faster quantum
  algorithm for solving systems of linear equations}, November 2010,
  arXiv:1010.4458.

\bibitem{AL22}
D.~An and L.~Lin, \emph{Quantum linear system solver based on time-optimal
  adiabatic quantum computing and quantum approximate optimization algorithm},
  ACM Transactions on Quantum Computing \textbf{3} (2022), no.~2, 1--28.

\bibitem{MR1648351}
I~Babu{\v{s}}ka and J.~E. Osborn, \emph{Can a finite element method perform
  arbitrarily badly?}, Math. Comp. \textbf{69} (2000), no.~230, 443--462.
  \MR{1648351}

\bibitem{Bac23}
M.~Bachmayr, \emph{Low-rank tensor methods for partial differential equations},
  Acta Numerica \textbf{32} (2023), 1--121.

\bibitem{BK20}
M.~Bachmayr and V.~Kazeev, \emph{Stability of {{Low-Rank Tensor
  Representations}} and {{Structured Multilevel Preconditioning}} for
  {{Elliptic PDEs}}}, Foundations of Computational Mathematics \textbf{20}
  (2020), no.~5, 1175--1236.

\bibitem{BNWA23a}
M.~Bagherimehrab, K.~Nakaji, N.~Wiebe, and A.~Aspuru-Guzik, \emph{Fast quantum
  algorithm for differential equations}, 2023, arXiv:2306.11802.

\bibitem{BBC+94}
R.~Barrett, M.~Berry, T.~F. Chan, J.~Demmel, J.~Donato, J.~Dongarra,
  V.~Eijkhout, R.~Pozo, C.~Romine, and H.~Van~der Vorst, \emph{Templates for
  the solution of linear systems: building blocks for iterative methods}, SIAM,
  1994.

\bibitem{Ber14}
D.~W. Berry, \emph{High-order quantum algorithm for solving linear differential
  equations}, Journal of Physics A: Mathematical and Theoretical \textbf{47}
  (2014), no.~10, 105301.

\bibitem{BCK15}
D.~W. Berry, A.~M. Childs, and R.~Kothari, \emph{Hamiltonian simulation with
  nearly optimal dependence on all parameters}, 2015 IEEE 56th annual symposium
  on foundations of computer science, IEEE, 2015, pp.~792--809.

\bibitem{BCOW17}
D.~W. Berry, A.~M. Childs, A.~Ostrander, and G.~Wang, \emph{Quantum algorithm
  for linear differential equations with exponentially improved dependence on
  precision}, Communications in Mathematical Physics \textbf{356} (2017),
  no.~3, 1057--1081.

\bibitem{BPX89}
J.~H. Bramble, J.~E. Pasciak, and J.~Xu, \emph{Parallel multilevel
  preconditioners}, Math. Comp. \textbf{55} (1990), no.~191, 1--22.
  \MR{1023042}

\bibitem{BH97}
G.~Brassard and P.~Hoyer, \emph{An exact quantum polynomial-time algorithm for
  simon's problem}, Proceedings of the Fifth Israeli Symposium on Theory of
  Computing and Systems, IEEE, 1997, pp.~12--23.

\bibitem{BS08}
S.~C. Brenner and L.~R. Scott, \emph{The mathematical theory of finite element
  methods}, third ed., Texts in Applied Mathematics, vol.~15, Springer, New
  York, 2008. \MR{2373954}

\bibitem{BCWD01}
H.~Buhrman, R.~Cleve, J.~Watrous, and R.~De~Wolf, \emph{Quantum
  fingerprinting}, Physical Review Letters \textbf{87} (2001), no.~16, 167902.

\bibitem{Bungartz_Griebel_2004}
H.~Bungartz and M.~Griebel, \emph{Sparse grids}, Acta Numerica \textbf{13}
  (2004), 147--269.

\bibitem{CBB+23}
Z.~Cai, R.~Babbush, S.~C. Benjamin, S.~Endo, W.~J. Huggins, Y.~Li, J.~R.
  McClean, and T.~E. O’Brien, \emph{Quantum error mitigation}, Reviews of
  Modern Physics \textbf{95} (2023), no.~4, 045005.

\bibitem{CPP+13}
Y.~Cao, A.~Papageorgiou, I.~Petras, J.~Traub, and S.~Kais, \emph{Quantum
  algorithm and circuit design solving the {{Poisson}} equation}, New Journal
  of Physics \textbf{15} (2013), no.~1, 013021.

\bibitem{CGJ19}
S.~Chakraborty, A.~Gily{\'e}n, and S.~Jeffery, \emph{The {{Power}} of
  {{Block-Encoded Matrix Powers}}: {{Improved Regression Techniques}} via
  {{Faster Hamiltonian Simulation}}}, LIPIcs, Volume 132, ICALP 2019
  \textbf{132} (2019), 33:1--33:14.

\bibitem{CKS17}
A.~M. Childs, R.~Kothari, and R.~D. Somma, \emph{Quantum {{Algorithm}} for
  {{Systems}} of {{Linear Equations}} with {{Exponentially Improved
  Dependence}} on {{Precision}}}, SIAM Journal on Computing \textbf{46} (2017),
  no.~6, 1920--1950.

\bibitem{CL20}
A.~M. Childs and J.~Liu, \emph{Quantum spectral methods for differential
  equations}, Communications in Mathematical Physics \textbf{375} (2020),
  no.~2, 1427--1457.

\bibitem{CLO21}
A.~M. Childs, J.~Liu, and A.~Ostrander, \emph{High-precision quantum algorithms
  for partial differential equations}, Quantum \textbf{5} (2021), 574.

\bibitem{CJS13}
B.~D. Clader, B.~C. Jacobs, and C.~R. Sprouse, \emph{Preconditioned quantum
  linear system algorithm}, Physical Review Letters \textbf{110} (2013),
  no.~25, 250504.

\bibitem{DK92}
W.~Dahmen and A.~Kunoth, \emph{Multilevel preconditioning}, Numerische
  Mathematik \textbf{63} (1992), no.~1, 315--344.

\bibitem{DP25_mfo}
M.~Deiml and D.~Peterseim, \emph{Quantum multilevel preconditioning}, to appear
  in Oberwolfach Rep., 2025.

\bibitem{DMN13}
S.~J. Devitt, W.~J. Munro, and K.~Nemoto, \emph{Quantum error correction for
  beginners}, Reports on Progress in Physics \textbf{76} (2013), no.~7, 076001.

\bibitem{Dra00}
T.~G. Draper, \emph{Addition on a {{Quantum Computer}}}, 2000,
  arXiv:quant-ph/0008033.

\bibitem{GSLW19}
A.~Gily{\'e}n, Y.~Su, Guang~H. Low, and N.~Wiebe, \emph{Quantum singular value
  transformation and beyond: Exponential improvements for quantum matrix
  arithmetics}, Proceedings of the 51st {{Annual ACM SIGACT Symposium}} on {{
  Theory}} of {{Computing}}, June 2019, pp.~193--204.

\bibitem{Gre66}
T.~N.~E. Greville, \emph{Note on the {{Generalized Inverse}} of a {{Matrix
  Product}}}, SIAM Review \textbf{8} (1966), no.~4, 518--521.

\bibitem{Griebel:1994*2}
M.~Griebel, \emph{Multilevel {{Algorithms Considered}} as {{Iterative Methods}}
  on {{ Semidefinite Systems}}}, SIAM Journal on Scientific Computing
  \textbf{15} (1994), no.~3, 547--565.

\bibitem{Griebel:1994*4}
\bysame, \emph{Multilevelmethoden als {{Iterationsverfahren}} {\"u}ber {{
  Erzeugendensystemen}}}, Teubner {{Skripten}} Zur {{Numerik}}, Vieweg+Teubner
  Verlag, Wiesbaden, 1994.

\bibitem{GR02}
L.~Grover and T.~Rudolph, \emph{Creating superpositions that correspond to
  efficiently integrable probability distributions}, 2002,
  arxiv:quant-ph/0208112.

\bibitem{Gro98}
L.~K. Grover, \emph{Quantum {{Computers Can Search Rapidly}} by {{Using Almost
  Any Transformation}}}, Physical Review Letters \textbf{80} (1998), no.~19,
  4329--4332.

\bibitem{HSS08}
H.~Harbrecht, R.~Schneider, and C.~Schwab, \emph{Multilevel frames for sparse
  tensor product spaces}, Numerische Mathematik \textbf{110} (2008), no.~2,
  199--220.

\bibitem{HHL09}
A.~W. Harrow, A.~Hassidim, and S.~Lloyd, \emph{Quantum algorithm for solving
  linear systems of equations}, Physical Review Letters \textbf{103} (2009),
  no.~15, 150502.

\bibitem{HW89}
C.~E. Heil and D.~F. Walnut, \emph{Continuous and {{Discrete Wavelet
  Transforms}}}, SIAM Review \textbf{31} (1989), no.~4, 628--666.

\bibitem{Hiptmair}
R.~Hiptmair, \emph{Operator preconditioning}, Comput. Math. Appl. \textbf{52}
  (2006), no.~5, 699--706. \MR{2275559}

\bibitem{HJLZ24}
J.~Hu, S.~Jin, N.~Liu, and L.~Zhang, \emph{Quantum {{Circuits}} for partial
  differential equations via {{Schr{\"o}dingerisation}}}, Quantum \textbf{8}
  (2024), 1563.

\bibitem{HJZ23}
J.~Hu, S.~Jin, and L.~Zhang, \emph{Quantum algorithms for multiscale partial
  differential equations}, Multiscale Modeling \& Simulation \textbf{22}
  (2024), no.~3, 1030--1067.

\bibitem{Qis23}
A.~Javadi-Abhari, M.~Treinish, K.~Krsulich, C.~J. Wood, J.~Lishman, J.~Gacon,
  S.~Martiel, P.~D. Nation, L.~S. Bishop, A.~W. Cross, B.~R. Johnson, and J.~M.
  Gambetta, \emph{Quantum computing with {Q}iskit}, 2024, arXiv:2405.08810.

\bibitem{PhysRevA.108.032603}
S.~Jin, N.~Liu, and Y.~Yu, \emph{Quantum simulation of partial differential
  equations: Applications and detailed analysis}, Phys. Rev. A \textbf{108}
  (2023), 032603.

\bibitem{KM04}
P.~Kaye and M.~Mosca, \emph{Quantum networks for generating arbitrary quantum
  states}, International Conference on Quantum Information, Optica Publishing
  Group, 2001, p.~PB28.

\bibitem{Lin22}
L.~Lin, \emph{Lecture notes on quantum algorithms for scientific computation},
  2022, arXiv:2201.08309.

\bibitem{LT20}
L.~Lin and Y.~Tong, \emph{Optimal polynomial based quantum eigenstate filtering
  with application to solving quantum linear systems}, Quantum \textbf{4}
  (2020), 361.

\bibitem{LOC24}
B.~Liu, M.~Ortiz, and F.~Cirak, \emph{Towards quantum computational mechanics},
  Computer Methods in Applied Mechanics and Engineering \textbf{432} (2024),
  117403.

\bibitem{Luo20}
A.~Luongo, \emph{Quantum algorithms for data analysis},
  https://quantumalgorithms.org, 2020.

\bibitem{Md23}
N.~S. Mande and R.~{de Wolf}, \emph{Tight {{Bounds}} for {{Quantum Phase
  Estimation}} and {{Related Problems}}}, LIPIcs, Volume 274, ESA 2023
  \textbf{274} (2023), 81:1--81:16.

\bibitem{MRTC21}
J.~M. Martyn, Z.~M. Rossi, A.~K. Tan, and I.~L. Chuang, \emph{Grand unification
  of quantum algorithms}, PRX Quantum \textbf{2} (2021), 040203.

\bibitem{MNU24}
S.~Mohr, Y.~Nakatsukasa, and C.~{Urz{\'u}a-Torres}, \emph{Full operator
  preconditioning and the accuracy of solving linear systems}, IMA Journal of
  Numerical Analysis \textbf{44} (2024), no.~6, 3259--3279.

\bibitem{MP16}
A.~Montanaro and S.~Pallister, \emph{Quantum algorithms and the finite element
  method}, Physical Review A \textbf{93} (2016), no.~3, 032324.

\bibitem{MBA+23}
S.~A. Moses, C.~H. Baldwin, M.~S. Allman, R.~Ancona, L.~Ascarrunz, C.~Barnes,
  J.~Bartolotta, B.~Bjork, P.~Blanchard, M.~Bohn, et~al., \emph{A race-track
  trapped-ion quantum processor}, Physical Review X \textbf{13} (2023), no.~4,
  041052.

\bibitem{doi:10.1137/18M1170650}
G.~Nannicini, \emph{An introduction to quantum computing, without the physics},
  SIAM Review \textbf{62} (2020), no.~4, 936--981.

\bibitem{Osw90}
P.~Oswald, \emph{On {{Function Spaces Related}} to {{Finite Element
  Approximation Theory}}}, Zeitschrift f{\"u}r Analysis und ihre Anwendungen
  \textbf{9} (1990), no.~1, 43--64.

\bibitem{Osw94}
\bysame, \emph{Multilevel {{Finite Element Approximation}}}, Teubner
  {{Skripten}} Zur {{Numerik}}, Vieweg+Teubner Verlag, Wiesbaden, 1994.

\bibitem{MR3022017}
D.~Peterseim and S.~Sauter, \emph{Finite {{Elements}} for {{Elliptic Problems}}
  with {{Highly Varying }}, {{Nonperiodic Diffusion Matrix}}}, Multiscale
  Modeling \& Simulation \textbf{10} (2012), no.~3, 665--695.

\bibitem{SWRK+23}
Y.~Sato, H.~C. Watanabe, R.~Raymond, R.~Kondo, K.~Wada, K.~Endo, M.~Sugawara,
  and N.~Yamamoto, \emph{Variational quantum algorithm for generalized
  eigenvalue problems and its application to the finite-element method},
  Physical Review A \textbf{108} (2023), no.~2, 022429.

\bibitem{MR1695813}
C.~Schwab, \emph{{$p$}- and {$hp$}-finite element methods}, Numerical
  Mathematics and Scientific Computation, The Clarendon Press, Oxford
  University Press, New York, 1998, Theory and applications in solid and fluid
  mechanics. \MR{1695813}

\bibitem{SX18}
C.~Shao and H.~Xiang, \emph{Quantum {{Circulant Preconditioner}} for {{Linear
  System}} of {{ Equations}}}, Physical Review A \textbf{98} (2018), no.~6,
  062321.

\bibitem{SSO19}
Y.~Suba{\c s}{\i}, R.~D. Somma, and D.~Orsucci, \emph{Quantum {{Algorithms}}
  for {{Systems}} of {{Linear Equations Inspired}} by {{Adiabatic Quantum
  Computing}}}, Physical Review Letters \textbf{122} (2019), no.~6, 060504.

\bibitem{SUR+20}
Y.~Suzuki, S.~Uno, R.~Raymond, T.~Tanaka, T.~Onodera, and N.~Yamamoto,
  \emph{Amplitude estimation without phase estimation}, Quantum Information
  Processing \textbf{19} (2020), 1--17.

\bibitem{TAWL21}
Y.~Tong, D.~An, N.~Wiebe, and L.~Lin, \emph{Fast inversion, preconditioned
  quantum linear system solvers, and fast evaluation of matrix functions},
  Physical Review A \textbf{104} (2021), no.~3, 032422.

\bibitem{TLDE23}
C.~J. Trahan, M.~Loveland, N.~Davis, and E.~Ellison, \emph{A variational
  quantum linear solver application to discrete finite-element methods},
  Entropy \textbf{25} (2023), no.~4, 580.

\bibitem{Vid03}
G.~Vidal, \emph{Efficient {{Classical Simulation}} of {{Slightly Entangled
  Quantum Computations}}}, Physical Review Letters \textbf{91} (2003), no.~14,
  147902.

\bibitem{Xu92}
J.~Xu, \emph{Iterative methods by space decomposition and subspace correction},
  SIAM Rev. \textbf{34} (1992), no.~4, 581--613. \MR{1193013}

\bibitem{Yueh05}
W.-C. Yueh, \emph{Eigenvalues of several tridiagonal matrices.}, Applied
  Mathematics E-Notes \textbf{5} (2005), 66--74.

\bibitem{Zal98}
C.~Zalka, \emph{Simulating quantum systems on a quantum computer}, Proceedings
  of the Royal Society of London. Series A: Mathematical, Physical and
  Engineering Sciences \textbf{454} (1998), no.~1969, 313--322.

\end{thebibliography}

\appendix

\section{Compressed encoding of finite element vectors} \label{app:implementation-vectors}

To solve a linear system on a quantum computer, we need to encode the right-hand side~$\ket{r}$ and possibly the vector $\ket{m}$ for which we want to measure $m^Tc$. We focus here on $r$, but the construction of $m$ follows exactly the same techniques. The way $r$ should be provided is as a quantum circuit $U_r$, which we require to map the zero state to $\ket{r}=\ket{r}_K$. Using the method independently discovered in \cite{GR02,KM04,Zal98}, such a circuit can be constructed from $K$ quantum circuits $U_r^{(1)},\ldots,U_r^{(K)}$, which realize the functions
\[x \in \{0, 1\}^k \mapsto g_k(x) \coloneqq \sum_{y \in \{0, 1\}^{K-k}} |r_{x \cdot y}|^2 \in \R\]
for $1 \le k \le K$. Here, the operation $\cdot$ denotes the concatenation of bit strings and $x \cdot y$ is to be understood as an integer index (indexing starts at zero). More precisely, $U_r^{(k)}$ implements the mapping
\[\ket{x}_{k-1}\ket{0} \mapsto \ket{x}_{k-1}\Bigg(\sqrt{\frac{g_k(x \cdot 0)}{g_{k-1}(x)}}\ket{0} + \sqrt{\frac{g_k(x \cdot 1)}{g_{k-1}(x)}}\ket{1}\Bigg).\]
The behavior of the mapping on the state $\ket{x}\ket{1}$ may be arbitrary. The same is true if the denominator $g_{k-1}(x)$ is zero.
If some of the components of $r$ are negative, we also need an additional oracle $U_r^\pm$ that flips the phase of the corresponding states. Then one can easily check that
\[U_r \coloneqq U_r^\pm \circ U_r^{(K)} \circ (I_1 \otimes U_r^{(K-1)}) \circ \dots \circ (I_{K-1} \otimes U_r^{(1)})\]
has the desired properties.

In the finite element context, we usually want to construct a vector representation $r$ of the bounded linear functional $f^*$ acting on finite element functions. This functional is given by its Riesz representation $f$ in the space $L^2$ and the components of $r$ are computed as $r_j = (\Lambda_j^{(L)}, f)_{L^2}$. Here $f$ is usually either provided as a symbolic formula, or constructed from some real world data such as measurements. It is reasonable to assume that these inner products can be computed efficiently and that the same is true for inner products with basis functions on coarser levels $(\Lambda_j^{(\ell)}, f)_{L^2}$ for $1 \le \ell \le L$ and $0 \le j \le (2^\ell - 1)^d$, but clearly the computation must scale at least linearly with the amount of data contained in $f$. For a typical use case, this will not change the asymptotic runtime of the algorithm, since the amount of data is almost always constant and independent of the discretization parameters.

The computation of the described inner products is (almost) sufficient to efficiently compute even the preconditioned vector $\ket{F^Tr}$. First, we describe the non-preconditioned case. There we see that $g_k$ can be computed as
\[ g_k(x) = \sum_{y\in\{0,1\}^{dL-k}} (\Lambda_{x \cdot y}^{(L)}, f)_{L^2} = (\Lambda_{k,x}^{(L)}, f)_{L^2}. \]
where for $1 \le \ell \le L$, $1 \le k \le d\ell$, and $0 \le x \le (2^k-1)^d$ we define
\[ \Lambda_{k,x}^{(\ell)} \coloneqq \sum_{y = 2^{d\ell-k}x}^{2^{d\ell-k}(x+1)-1}\Lambda_y^{(\ell)}. \]
A visualization of this function can be seen in Figure \ref{fig:basis-function-sums}. It is unclear whether the inner products with these functions could be computed efficiently from the inner products with the basis functions, but one can see that each $\Lambda_{k,x}^{(\ell)}$ is a sum of ``half'' basis functions, and we present a visual proof in Figure \ref{fig:basis-function-sums} in the one-dimensional case.

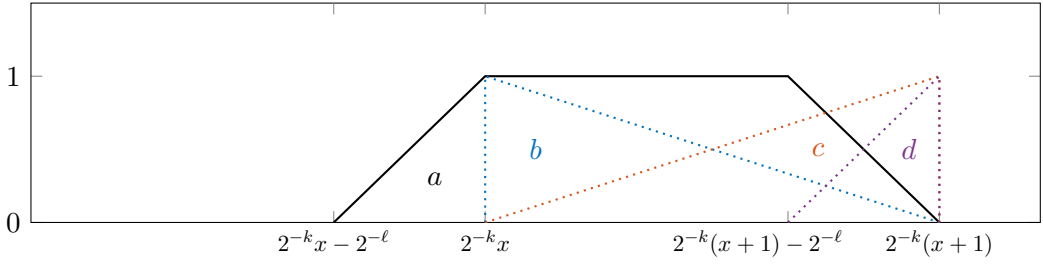
\begin{figure}
\begin{tikzpicture}
    \begin{axis}[
    width = \textwidth,
    height = 0.3\textwidth,
    xmin = 0,
    xmax = 1,
    ymin = 0,
    ymax = 1.5,
    ytick = {0, 1},
    xtick = \empty,
    extra x tick style={tick label style={scale=0.8}},
    extra x ticks  = {0.3,0.45,0.75,0.9},
    extra x tick labels  = {$2^{-k}x - 2^{-\ell}$, $2^{-k}x$,$2^{-k}(x+1)-2^{-\ell}$\hspace*{1cm},$2^{-k}(x+1)$},
    ]
    \addplot[thick] coordinates {
        (0.3, 0)
        (0.45, 1)
        (0.75, 1)
        (0.9, 0)
    };
    \addplot[mycolor1, thick, dotted] coordinates {
        (0.45, 0)
        (0.45, 1)
        (0.9, 0)
    };
    \addplot[mycolor2, thick, dotted] coordinates {
        (0.45, 0)
        (0.9, 1)
        (0.9, 0)
    };
    \addplot[mycolor3, thick, dotted] coordinates {
        (0.75, 0)
        (0.9, 1)
        (0.9, 0)
    };
	\node [] at (axis cs:0.4,0.3) {$a$};
	\node [mycolor1] at (axis cs:0.5,0.5) {$b$};
	\node [mycolor2] at (axis cs:0.78,0.5) {$c$};
	\node [mycolor3] at (axis cs:0.87,0.5) {$d$};
    \end{axis}
    
\end{tikzpicture}
\caption{ \label{fig:basis-function-sums}
Visualization of the function $\Lambda_{k,x}^{(\ell)}$ as the sum of $4$ ``half'' basis functions $a + b + c - d$.
}
\end{figure}

For the preconditioned case, the preparation of the states for the right hand side is more complicated due to the fact that we have to assemble the preconditioned state $\ket{F^Tr}$ directly. For simplicity, let us assume that there is some $\tilde L \in \N$ such that $L = 2^{\tilde L}$. For the computation of $g_k$ we then consider two cases separately. First, for $1 \le k \le \tilde L$ the variable $x$ in each $g_k(x)$ refers only to the grid level $\ell$ on which information is acquired. Instead of giving an explicit formula for $g_k$, it is sufficient to understand that the combined transformation $O_r^{(\tilde L)} \circ \dots \circ O_r^{(1)}$ is given by
\[ \ket{0}_{\tilde L} \mapsto \frac{1}{\omega} \sum_{\ell=1}^L 2^{-\ell}(\Lambda_{\ell,0}^{(\ell)}, f)_{L^2} \ket{\ell} \]
where $\omega \coloneqq (\sum_{\ell=1}^L 2^{-2\ell}(\Lambda_{\ell,0}^{(\ell)}, f)_{L^2}^2)^{-1/2}$ is a normalization factor. The total time to compute this is $\bigO(L)$, since it is a sum of $L$ summands, each of which can be computed in $\bigO(1)$ time.

For $\tilde L + 1 \le k \le dL + \tilde L$ on the other hand we interpret $x$ as consisting of the level $\ell$ and a spatial part $z$. The computation is almost the same as in the non-preconditioned case, differing only in the diagonal scaling of the preconditioner
\[ g_k(\ell, z) = 2^{-\ell}(\Lambda_{k-\tilde L,z}^{(\ell)}, f)_{L^2} \]
where the circuits will behave correctly if we artificially define $g_{\tilde L} \equiv 1$. Obviously, from the previous considerations, this can also be computed in time $\bigO(1)$, i.e.\ for the entire encoding procedure we obtain a runtime of $\bigO(L)$.

It has been shown that polynomials satisfy the requirements of this construction \cite{MP16} and as such can be efficiently encoded. Similarly, one can check that the same is true for the basis functions $\Lambda_j^{(\ell)}$ and linear combinations thereof. However, a detailed analysis of the kind of functions that can be encoded does not exist to the best of our knowledge.

\section{Proof of Proposition~\ref{prop:block-encoding-ops}}\label{app:proof}

\begin{figure}[p]
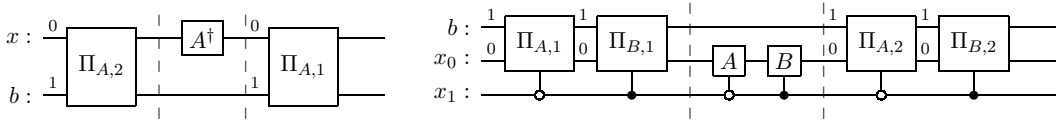
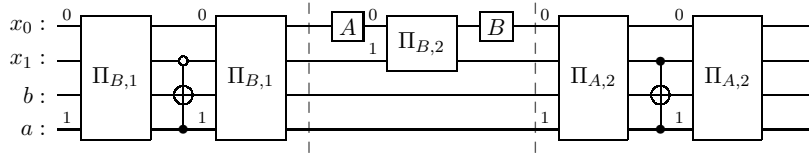
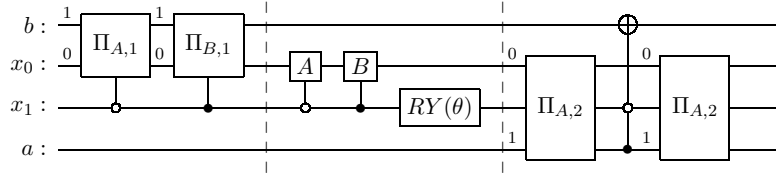
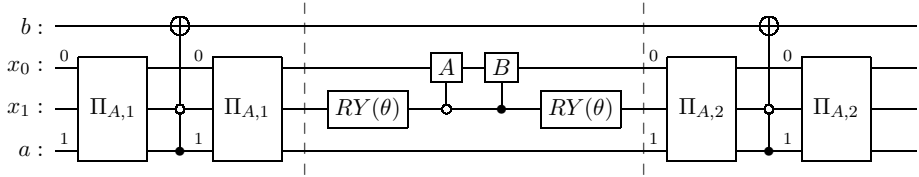

\begin{subfigure}[b]{\textwidth}
    \centering
    \includestandalone{circuits/ops/tensor}
    \caption{Block encoding for $A \otimes B$}
\end{subfigure}
\begin{subfigure}[b]{.34\textwidth}
    \includestandalone{circuits/ops/transpose}
    \hspace{-4ex}
    \vspace{-4.6ex}
    \subcaption{Block encoding for $A^\dagger$\vphantom{$\begin{bmatrix}A & 0 \\ 0 & B \end{bmatrix}$}}
\end{subfigure}
\begin{subfigure}[b]{.65\textwidth}
    \centering
    \includestandalone{circuits/ops/block_diagonal}
    \subcaption{Block encoding for $\begin{bmatrix}A & 0 \\ 0 & B \end{bmatrix}$}
\end{subfigure}
\begin{subfigure}[b]{\textwidth}
    \centering
    \includestandalone{circuits/ops/multiply}
    \caption{Block encoding for $AB$}
\end{subfigure}
\begin{subfigure}[b]{\textwidth}
    \centering
    \includestandalone{circuits/ops/block_horizontal}
    \caption{Block encoding for $\begin{bmatrix}A & B \end{bmatrix}$ where $\theta = \arctan(\gamma_B / \gamma_A)$}
\end{subfigure}
\begin{subfigure}[b]{\textwidth}
    \centering
    \includestandalone{circuits/ops/add}
    \caption{Block encoding for $\mu_A A + \mu_B B$ where $\theta = \arctan(\sqrt{\gamma_B\mu_B / (\gamma_A\mu_A)})$}
\end{subfigure}
\caption{
\label{fig:ops-circuits}
Circuits of block encoding for the operations of Proposition~\ref{prop:block-encoding-ops}. Implementations for both the unitary and corresponding $\CNOT{\Pi}$ gates are shown.
}
\end{figure}
To prove the proposition, Figure~\ref{fig:ops-circuits} gives quantum circuits that implement the block encodings for each operation with the exact normalization $\gamma$. Additionally, the bounds on the subnormalization $\tilde \gamma$ follow from lower bounds on the norm of the resulting matrix as follows.

     $A \otimes B$:
    An easy calculation shows that
    \begin{align*}
    &\gamma_A \gamma_B (\Pi_{A,2} \otimes \Pi_{B,2}) (U_A \otimes U_B) (\Pi_{A,1} \otimes \Pi_{B,1})^\dagger \\
    &\qquad=
    (\gamma_A \Pi_{A,2} U_A \Pi_{A,1}^\dagger) \otimes (\gamma_B \Pi_{B,2} U_B \Pi_{B,1}^\dagger)
    =
    A \otimes B
    \end{align*}
    and we can exactly calculate the norm as $|A \otimes B| = |A|\, |B|$. For constructing the $\CNOT{\Pi}$ operations we use that a state is in the range of a projection $\Pi_{A,\bullet} \otimes \Pi_{B,\bullet}$ exactly if it is in the range of $\Pi_{A,\bullet} \otimes \Id$ and in the range of $\Id \otimes \Pi_{B,\bullet}$. For computing the logical AND of two boolean values in a quantum computer, the Toffoli gate is used. The two original values need to be uncomputed subsequently.
    
     $A^\dagger$:
    An easy calculation shows that
    \begin{align*}
    &\gamma_A \Pi_{A,1} U_A^\dagger \Pi_{A,2}^\dagger =
    (\gamma_A \Pi_{A,2} U_A \Pi_{A,1}^\dagger)^\dagger
    =
    A^\dagger
    \end{align*}
    and we can exactly calculate the norm as $|A^\dagger| = |A|$.
    
     $\begin{bmatrix} A & 0 \\ 0 & B\end{bmatrix} \eqqcolon D$:
    Let us first assume that the block encodings for both $A$ and $B$ have the same normalization $\gamma_A = \gamma_B$. Then we can calculate
    \begin{align*}
    &\gamma_A (\Pi_{A,2} \otimes \op00 + \Pi_{B,2} \otimes \op11)(U_A \otimes \op00 + U_B \otimes \op11) \\
    &\qquad (\Pi_{A,1} \otimes \op00 + \Pi_{B,1} \otimes \op11)^\dagger \\
    &\qquad =
    (\gamma_A \Pi_{A,2} U_A \Pi_{A,1}^\dagger) \otimes \op00 + (\gamma_B \Pi_{A,2} U_A \Pi_{A,1}^\dagger) \otimes \op11 = D
    \end{align*}
    If the encodings do not have the same normalization, then it is possible to increase the lower of the two normalizations, though this needs an additional qubit. For the matrix norm we know that $|D| = \max\{|A|, |B|\}$.
    
     $AB$:
    Indeed the construction works without the extra requirement on the size of the matrices, which instead is only needed to bound the subnormalization. The validity of the encoding follows from
    \begin{align*}
    &\gamma_A\gamma_B (\Pi_{A,2} \otimes \bra{1}) (U_A \otimes \Id_2) \CNOT{\Pi_{A,1}} (U_B \otimes \Id_2) (\Pi_{B,1} \otimes \bra{0})^\dagger \\
    &\qquad=
    \gamma_A\gamma_B \Pi_{A,2} U_A (\Id \otimes \bra{1}) \CNOT{\Pi_{A,1}} (\Id \otimes \ket{0}) U_B \Pi_{B,1}^\dagger \\
    &\qquad=
    (\gamma_A \Pi_{A,2} U_A \Pi_{A,1}^\dagger) (\gamma_B \Pi_{B,2} U_B \Pi_{B,1}^\dagger)
    = AB
    \end{align*}
    where we used that by definition $(\Id \otimes \bra{1}) \CNOT{\Pi_{A,1}} (\Id \otimes \ket{0}) = \Pi_{A,1}^\dagger\Pi_{A,1}$. The projections $\Pi_\bullet \otimes \bra{0}$ can be constructed similar to the ones in the $A \otimes B$ case. Finally we need a lower bound on the matrix norm of the product. For this consider
    \[ |A|\, |B| \le |AB| \, |B^+|\, |B| = \kappa(B) |AB| \]
    where we have to assume that $n_B \ge m_B$ and equivalently we get $\kappa(A)|AB| \ge |A|\, |B|$ if $m_A \ge n_A$. The bounds for the subnormalization follow by straightforward calculation.
    
    $\begin{bmatrix} A & B\end{bmatrix}$:
    This matrix can be constructed from a block diagonal using the equality
    \[ \begin{bmatrix} A & B \end{bmatrix} = \big(\begin{bmatrix} \gamma_A & \gamma_B \end{bmatrix} \otimes \Id\big) \begin{bmatrix} A/\gamma_A & 0 \\ 0 & B/\gamma_B \end{bmatrix} \]
    where $\begin{bmatrix} \gamma_A & \gamma_B \end{bmatrix}$ is encoded by a $Y$-rotation with normalization $\sqrt{\gamma_A^2 + \gamma_B^2}$ and the block square matrix due to the scaling has normalization $1$. The subnormalization can be calculated using
    \[ \big|\begin{bmatrix} A & B \end{bmatrix}\big| \le \sqrt{|A|^2 + |B|^2} \]

     $\mu_A A + \mu_B B$:
    This can similarly be constructed from a block diagonal by
    \[ \mu_A A + \mu_B B = \big(\begin{bmatrix} \sqrt{\gamma_A\mu_A} & \sqrt{\gamma_B\mu_B} \end{bmatrix} \otimes \Id\big) \begin{bmatrix} A/\gamma_A & 0 \\ 0 & B/\gamma_B \end{bmatrix} \bigg(\begin{bmatrix} \sqrt{\gamma_A\mu_A} \\ \sqrt{\gamma_B\mu_B} \end{bmatrix}\otimes \Id\bigg) \]
    where the left and right vectors lead to a total normalization of $\gamma_A\mu_A + \gamma_B\mu_B$. As mentioned, the matrix norm of a sum can in general not be bounded from below.

\end{document}